\documentclass[
  aps,
  journal=pra,         % required by REVTeX
  onecolumn,
  nofootinbib
]{revtex4-2}

\usepackage{amsmath,amssymb,amsfonts,bm,physics}
\usepackage{mathtools,amsthm}
\usepackage{hyperref}
\usepackage[capitalize]{cleveref}
\usepackage{mathrsfs}
\usepackage{geometry}
\usepackage{bm}
\newcommand{\Op}{\operatorname{Op}}
\usepackage{tikz}
\usepackage{tikz-cd}        % ← this provides \begin{tikzcd} ... \end{tikzcd}
\tikzset{>=latex}           % nice arrow heads

\usepackage{enumitem}

\newcommand{\co}{\mathcal{O}}
\newcommand{\ch}{\mathcal{H}}
\newcommand{\ct}{\mathcal{T}}

\newtheorem{example}{Example}

\newcommand{\dt}{\dfrac{d}{dt} }

\newcommand{\brr}[1]{({#1}\rvert}
\newcommand{\kee}[1]{\lvert{#1})}

\newtheorem{theorem}{Theorem}
\newtheorem{lemma}{Lemma}

\newtheorem{corollary}{Corollary}

\usepackage{tikz}
\usetikzlibrary{shapes, arrows.meta, positioning, calc, decorations.pathreplacing} % <--- Added decorations.pathreplacing

\begin{document}
\title{From Linear Differential Equations to Unitaries: A Moment-Matching Dilation Framework with Near-Optimal Quantum Algorithms }
\author{Xiantao Li, \\ The Pennsylvania State University, \\
Xiantao.Li@psu.edu }

\begin{abstract}

Quantum speed-ups for dynamical simulation usually demand \emph{unitary} time-evolution, whereas the large ODE/PDE systems encountered in realistic physical models are generically \emph{non-unitary}.  
We present a universal \textit{moment-fulfilling} dilation that embeds any linear, non-Hermitian differential equation  
$\dot x = L(t) x$ with the generator $L=-iH+K$ into a strictly unitary evolution on an enlarged Hilbert space:  
\[
   \Big(\brr{l}   \otimes I_A\Big)  \displaystyle
   \mathcal T e^{-i   \int ( I_A \otimes H + i F\otimes K)  dt}  
   \Big(\kee{r}   \otimes I_A\Big)
   = \mathcal T e^{\int L  dt},
\]
provided the triple $(F,\kee{r},\brr{l})$ satisfies the compact
moment identities $\brr{l}F^{k}\kee{r}=1$ for all $k\ge 0$ in the ancilla space.  
This algebraic criterion recovers both \emph{Schr\"odingerization}
[Phys.\ Rev.\ Lett.\ \textbf{133}, 230602 (2024)] and the
linear-combination-of-Hamiltonians (LCHS) scheme
[Phys.\ Rev.\ Lett.\ \textbf{131}, 150603 (2023)], while also
unveiling whole \emph{families} of new dilations built from
differential, integral, pseudo-differential, and difference
generators. Each family comes with a continuous tuning parameter \emph{and} is
closed under similarity transformations that leave the moments
invariant, giving rise to an overwhelming landscape of design space 
that allows quantum dilations to be
co-optimized for specific applications, algorithms, and hardware.

As concrete demonstrations, we prove that a simple
finite-difference dilation in a finite interval attains near-optimal oracle complexity, and we also construct a Bargmann–Fock dilation for
continuous-variable platforms with a single bosonic mode.  Numerical experiments on Maxwell viscoelastic wave propagation confirm the
accuracy and robustness of the approach.

\end{abstract}

\maketitle

\section{Introduction}

Large–scale simulations based on ordinary and partial differential
equations (ODEs/PDEs) are indispensable to modern, computation-driven
physical design and prediction.  
Their principal bottleneck is \emph{dimension}: realistic models easily
involve a huge number of coupled degrees of freedom, making further physical
refinement prohibitively expensive on classical hardware.
One class of differential equations that can be simulated on quantum computers with a promised exponential speedup is the time–dependent Schr\"odinger equations
(TDSEs), for which many efficient algorithms, known as Hamiltonian simulation algorithms,   have been developed with optimal scaling in precision and simulation time \cite{LC16,low2018hamiltonian,gilyen2019quantum}. 

Most differential equations of practical interest, however, are
\emph{non-unitary}: relaxation, dissipation or gain render the
generator $L$ non-Hermitian on its Hilbert space $\ch$.  
To harness Hamiltonian-simulation techniques, one must therefore
\emph{dilate} the flow, i.e.\ embed it into a larger Hilbert space
on which the evolution becomes unitary.
A natural starting point is the decomposition
$L=-iH+K$ with $H=H^{\dagger}$ and $K=K^{\dagger}$.
Two pioneering dilations are the Schrödingerization method \cite{PhysRevLett.133.230602,jin2023quantum}, which uses a warped-phase ancillary  coordinate, and the linear combination of Hamiltonian simulations (LCHS) \cite{ALL23}, which 
realizes the propagator as a Fourier integral. 

Both schemes deliver an \emph{exact} embedding, and after a suitable
numerical discretization, inheriting the exponential speedup in Hilbert-space
dimension together with near–optimal scaling in the evolution time and
target precision~$\epsilon$ \cite{ACL23,jin2025schr}.
Yet they rely on very different ingredients in the implementation.
A recent note \cite{hu2025dilation} further points out that Schrödingerization can be
interpreted through the lens of the general
Sz.–Nagy dilation theorem 
\cite{nagy2010harmonic} as a specific realization.
The confluence of these rapid developments, and the clear momentum behind their initial applications \cite{jin2024quantum,bharadwaj2025compact,jin2024quantum-interf,hu2024quantum}  brings forth three pressing questions into sharp focus: \textbf{(1)} \emph{Is there a unifying abstraction that embraces
      \emph{both} existing dilations \textbf{and} makes explicit the
      freedom in choosing the ancillary space and its coupling with the system Hilbert space $\ch$?}
\textbf{(2)} \emph{Which structural criteria on a dilation
      guarantee an \emph{exact} embedding of the non-unitary dynamics?}
\textbf{(3)} \emph{Can those criteria be leveraged to
      \emph{systematically} generate new, tunable families of dilations
      tailored to specific algorithms and hardware—
      from digital qubits to continuous-variable photonics
      and neutral-atom platforms?}

The present work answers all three questions in the affirmative. Specifically, we present a  general framework for treating any linear PDE or ODE system with generator $L(t)$, where we first streamline the dilation into an Encode-Evolve-Evaluate  pipeline:
\[
   \underbrace{\brr{l} \otimes I \bigr)}_{\text{\bf E}valuate} 
         \underbrace{\ct e^{\displaystyle  -i   \int_{0}^{t}\bigl(I_A \otimes H(s) +i F \otimes K(s)\bigr)ds }
       }_{\text{\bf E}volve} 
   \underbrace{\bigl(\kee{r} \otimes I }_{\text{\bf E}ncode}   \longrightarrow    \ct e^{\int_0^t L(s)ds}  
\]
where $I_A$ and $F$  act  on an ancillary  Hilbert space
$\mathcal H_A$. $F$ is skew Hermitian, thus making the dilated operator $I_A \otimes H(s) +i F \otimes K(s)$ Hermitian on $\ch_A \otimes \ch$ .
Meanwhile 
$\kee{r}\in\mathbb X$ with $\mathbb X$ being a larger function space than $\mathcal H_A$, but approximitable by $\ch_A$. The notation $\kee{\cdot}$ indicates that the state might not be normalized. Similarly, 
$\brr{l}$ is a linear functional on $\mathbb X$. This procedure is inspired by the Schrödingerization method—formulated with warped-phase coordinates and a Fourier transform \cite{PhysRevLett.133.230602} or in the qumode-quadrature basis \cite{jin2024analog}. By setting up a general framework that does not presuppose a particular ancilla space, we preserve the flexibility to devise alternative dilation schemes, as demonstrated below.

\smallskip 

The main contribution of this paper can be summarized as follows
\begin{enumerate}
    \item We first show that the dilation is exact, provided that the following moment-matching condition is satisfied:
\(
   \brr{l} F^{k}\kee{r}=1 (k\ge0),
\)
thus providing a set of transparent conditions.

\item We demonstrate that the moment framework immediately reproduces the two
classical dilations (Schr\"odingerization and LCHS) and
generates many more. We list five of them in \cref{tab:dilation-families}. Each of these choices is expressed as a family of moment-fulfilling dilation schemes with a parameter $\theta$ to optimize the simulation complexity. 

\begin{table}[htp]
\centering
\renewcommand{\arraystretch}{1.25}
\begin{tabular}{@{}l l c c c@{}}
\hline\hline
Space $H_A$ & Dilation operators & Generator $F$ & Right vector $\kee{r}$ & Evaluation $\brr{l}$ \\ \hline
 $H^1(0,1)\quad $  & Differential  &
$\displaystyle \theta   \left(p\partial_{p}+\tfrac12\right)$ &
$\displaystyle p^{1/\theta-1/2}$ &
$  \displaystyle \brr{l}f= 2^{1/\theta-1/2}  f(\tfrac12)  $ \\[2pt]

$L^2(\mathbb R)$  &Integral &
$\displaystyle (Ff)(p)=      \int_{\mathbb R}    p  e^{-\theta|p-q|}f(q)  dq \quad $ &
$e^{a(\theta)p}$  &
$\displaystyle \brr{l}f=f(0)$ \\[4pt]

$L^2(\mathbb R)$  &Pseudo-differential &
$\displaystyle -i  (-\Delta)^{\theta}$ &
$e^{i\xi_{0}x}$, $ \xi_{0}=e^{i\pi/(4\theta)}$ &
$\displaystyle \brr{l}f=f(0)$ \\[4pt]

$\mathscr B$  & Bargmann–Fock &
$\displaystyle \theta\bigl(a^{\dagger}-a\bigr)$ &
$\displaystyle \exp      \Bigl(\frac{z^{2}}{2}-\frac{z}{\theta}\Bigr)$ &
$\displaystyle \brr{l}f=f(0)$ \\ 

$\ell^2$  &  Difference &
$\displaystyle (F f)_n= \theta (f_n - f_{n-1}) $ &
$\{\lambda_\theta^n\}_{n\geq 0},  \lambda=\dfrac{\theta}{1+\theta} $   &
$\displaystyle \brr{l}f=f_0$ \\[4pt]

\hline\hline
\end{tabular}
\caption{New one–parameter families of exact moment-dilations.
The free parameter $\theta$ 
tunes both the ancillary  generator $F$ and the vector
$\kee{r}$; the moment conditions \cref{moms} are
satisfied in every case.}
\label{tab:dilation-families}
\end{table}

\item We analyze the dilation scheme based on the differential operator on a compact domain $F=(p\partial_{p}+\tfrac12)$ on $(0,1)$. Using simple finite-difference discretizations, we show how near-optimal quantum complexity can be easily achieved.  

\item Using $F=a^{\dagger}-a$, we also construct the $\brr{l}$ and $\kee{r}$ pair that fulfills the moment conditions. By working with Bargmann–Fock  space, we demonstrate the feasibility of the dilation scheme for continuous variable quantum computing platforms. 
\end{enumerate}

A direct consequence of our general dilation scheme is a \emph{universal approximation property} for solution trajectories of arbitrary linear ODEs. Specifically, for any linear PDE/ODE $\dot{\bm x}(t)=L(t)\bm x(t)$ and any horizon $T>0$, there exist a unitary flow $U(t)$ on a dilated Hilbert space and fixed boundary vectors $\bra{l}$, $\ket{r}$ such that
\(
(\bra{l} \otimes I) U(t) (\ket{r} \otimes \bm x_0)
\)
approximates $\bm x(t)$ uniformly on $[0,T]$ to arbitrary accuracy. This is achieved by selecting a sequence $\ket{r^{(n)}} \in \mathcal H_A$ converging to the target continuous profile and the corresponding duals $\bra{l^{(n)}}$, so that the projected unitary trajectories converge to the ODE solution. Notably, the construction is agnostic to spectral or contractive properties of $L(t)$: it applies equally to dissipative, conservative, and transiently or persistently unstable dynamics.

\smallskip

The framework thus bridges current algorithms,
offers a catalogue of ready-to-use dilation, and suggests many more yet to
be explored by the quantum-algorithm community. Notably, beyond dilation schemes, most of the quantum algorithms have been based on quantum linear solvers \cite{Ber14,BCOW17,BC22,childs2020quantum,krovi2023improved,dong2025quantum,wu2024structure}, which also produces a quantum state proportional to the solution $\bm x(T),$ but requires significantly more state preparation of $\ket{\bm x_0}$. Dilation schemes can be viewed as a time-marching scheme, but unlike the method \cite{FLT22}, which requires frequent amplitude amplification, dilation schemes wrap the dynamics into a unitary one, and the amplitude amplification is only applied at the final time.     

As a specific example, we will analyze the dilation scheme based on a differential operator, which after the spatial discretization, leads to a tight-binding model on the ancilla with nearest-neighbor hoppings. We show that the resulting quantum algorithm exhibits near-optimal complexity in all parameters. 

\begin{theorem}[Informal main result]
Let $\dot{\bm x}(t)=(-iH(t)+K(t))\bm x(t)$ with $H(t)=H(t)^\dagger$ and $K(t)=K(t)^\dagger$ (allowing gain).  
Using the differential generator $F=p\partial_p+\tfrac12$ and a moment-matching ancilla triple $(F,\ket{r},\bra{l})$, we construct an exact unitary dilation
$\widetilde H(t)=I_A\!\otimes H(t)+i\,\theta F\otimes K(t)$ whose ancilla readout recovers $\bm x(t)$.
After discretization, $F$ becomes a nearest-neighbor tight-binding chain $F_h$, giving a 2-local qubit implementation and a light-cone suppression of boundary error.  
With a segment-wise evolution and oblivious amplitude amplification, this yields a near-optimal query complexity
$\widetilde{\mathcal O}\!\big(T(H_{\max}+K_{\max})\,\|\bm x(0)\|/\|\bm x(T)\|\big)$ (up to standard polylog factors).
In the autonomous case ($H(t)\equiv H$, $K(t)\equiv K$), combining segmentation with optimal Hamiltonian simulation yields only an \emph{additive} logarithmic dependence on precision (of size $\mathcal O(\|K\|T\log(\|K\|T/\epsilon))$), and recovers the optimal Hamiltonian-simulation scaling as $\|K\|\to 0$.
\end{theorem}

The remainder of the paper is organized as follows.
In \cref{sec:setup} we introduce the general moment-fulfilling dilation framework and survey the broad landscape of admissible triples $\bigl(F,\kee{r},\brr{l}\bigr)$.
\cref{sec:compact-dilation} specializes the framework to a finite-interval, finite-difference dilation that attains nearly optimal scaling in the precision $\epsilon$, evolution time $T$, and oracle queries to $L(t)$; the section concludes with a numerical demonstration on Maxwell viscoelastic waves.
Finally, \cref{sec:bargmann} presents a Bargmann–Fock realization, showcasing that  the scheme is amenable to continuous-variable quantum hardware.

\section{Problem setup and mathematical framework}
\label{sec:setup}

We want to simulate the linear systems of $N$-dimensional ordinary differential equations, or partial differential equations after spatial discretization,
\begin{equation}\label{eq:phys-ode}
   \dt \bm x(t) =   L(t) \bm x(t), \bm x(0)=\bm x_{0}\in\mathbb C^{N}.
\end{equation}
Here $L \in \mathbb C^{N\times N} $. Due to the linearity, one can assume that $\norm{\bm x_0}=1$ without loss of generality.  The goal of a differential equation solver is to approximate $\bm x(t)= \ct e^{\int_0^t L(s)ds} \bm x_0.$

Similar to \cite{PhysRevLett.133.230602,jin2023quantum,ALL23},  our approach begins with a partition of $L$
\begin{equation}\label{A2HK}
    L = -i H + K,  
\end{equation}
with Hermitian matrices $H$ and $K$.  Due to the matrix $K$, the dynamics of \cref{eq:phys-ode} is non-unitary, which can not be directly treated with Hamiltonian simulation techniques.  

We focus on dilation approaches that turn~\eqref{eq:phys-ode} into a
Schr\"odinger equation on an enlarged Hilbert space and is organized in
three, as an \textbf{E–E–E} pipeline  
(Encode $\to$ Evolve $\to$ Evaluate):
\begin{enumerate}
    \item 
an {\bf E}ncoding step that maps the initial condition in \cref{eq:phys-ode} to a quantum state: $\ket{\Psi(0)} = \kee{r} \ket{\bm x_0} $. Here we use the notation  $\kee{r}$ to highlight a vector that is not in the Hilbert space. 
    \item an {\bf E}volving step, where we simulate a unitary dynamics 
    \begin{equation}\label{unitary}
        \ket{\Psi(t)} = \displaystyle \ct e^{-i \int_0^t \widetilde{H}(s)ds}\ket{\Psi(0)},
    \end{equation}
    with a dilated Hamiltonian,
   \begin{equation}\label{eq: diH}
       \widetilde{H}(t):=   I_A \otimes H(t)
         +
         i F \otimes K(t), 
   \end{equation} 
   with some appropriate skew-Hermitian operator $F$ in a Hilbert space $\ch_A$. 
   \item an {\bf E}valuating step that recovers the solution of the original non-unitary dynamics in \cref{eq:phys-ode}: $\bm x(t) =  \brr{l} \otimes I  \ket{\Psi(t)},$ with a suitable linear functional $\brr{l}$. 
          
\end{enumerate}

These three steps can be concisely summarized in the following identity,
\begin{equation}\label{eq: schrodingerization0}
           \ct e^{\int_0^t L(s)ds} = \big( \brr{l} \otimes I  \big)  
       \ct e^{-i \int_0^t\widetilde H(s) ds}   
       \big( \kee{r} \otimes I \big), 
        \qquad
        \forall  t\ge 0.
\end{equation}
Here $I$ is the identity operator on the original space $\ch.$

To formulate a route that \emph{exactly}  embeds the non-unitary dynamics into the unitary dynamics \cref{unitary}, it is necessary to place it into an infinite-dimensional function space. Specifically,  we let $\mathbb{X} $ be a metric function space with $\kee{r} \in \mathbb{X}$, and we let $\mathcal H_A$ be the ancilla Hilbert space:  $F \in \ch$. We require that $\mathcal H_A  \subset \mathbb{X}$ but $\ch_A$ is dense in $\mathbb{X}$ so that $F$ can be extended to functions in $\mathbb{X}$. Finally, $\brr{l}$ is a  linear functional on $\mathbb{X}$.

\begin{theorem}[Moment-fulfilling and Exact Dilation]\label{thm:moment-recovery}
   Let $L(t)=-i H(t)+K(t) $ be split as in \cref{A2HK}. 
   Then the dilated evolution and recovery
   \emph{exactly} reproduces the physical solution as in \eqref{eq: schrodingerization0},
    provided that the following moment conditions are fulfilled,
\begin{equation}\label{moms}
    \brr{l} F^k \kee{r} =1, \quad \forall k\geq 0. 
\end{equation}
\end{theorem}

\begin{proof}
 For autonomous ODEs, $A$ is time-independent.  We can express the exponential as a Taylor series, each term can be expanded with powers: $\bigl( I_A \otimes H+ i F\otimes K \bigr)^k$. When taking the moments of each term, $\brr{l}$ and $\kee{r}$ will only act on $F$, which then is reduced to the Taylor expansion of the matrix exponential of $-iH+K$, and thus completes the proof. 

    To prove the time-dependent case, we can resort to the time-ordered evolution operator, which involves the time-ordered integrals of the following form,
    \[
   \brr{l}   \bigl( I_A \otimes   H(t_1)
         +
         i  F \otimes K(t_1)  \bigr) 
          \bigl(  I_A \otimes   H(t_2)
         +
         i F \otimes K(t_2)   \bigr) 
         \cdots 
          \bigl( I_A \otimes   H(t_k)
         +
         i F \otimes K(t_k)  \bigr) \kee{r}.
    \]
Multiplying the operators and using the fact that $\brr{l} I_A \kee{r} =1$ and the properties in \cref{moms}, we find that the product is reduced to,
\[
  \bigl( H(t_1) +  i K(t_1)   \bigr) 
          \bigl( H(t_2)  +
         i K(t_2)  \bigr) 
         \cdots 
          \bigl( H(t_k) +
         i K(t_k)   \bigr),
\]
leading exactly to the Dyson series expansion for the time-ordered evolution operator associated with $L(t)$.

\end{proof}

 It is crucial to point out that  working entirely in a Hilbert space will never fulfill the moment conditions \cref{moms}.  
If both $\brr{l}, \kee{r} \in \ch_A$, then one can see the following contradiction:
the right hand side of \cref{eq: schrodingerization0} is bounded for all $t$, whereas the  left hand side has no such guarantee.   It is precisely this use of states outside the ancilla Hilbert space that grants the dilation its power. One important implication of the current result is that the trajectories from \eqref{eq: schrodingerization0}  are dense in the solution space of any linear ODE system, including those with transient or persistent solution growth, thus achieving a theoretical universal approximation theorem using time-dependent Schr\"odinger equations. This can be seen by choosing an approximation of $\kee{r}$ in $\ch$, with which the dynamics generated by $\widetilde H(t)$ is unitary.

\smallskip 
A simple recipe to fulfill the moment conditions \cref{moms} is to ensure $F\kee{r}=\kee{r}$ and $( l|r)=1$. Since $F$ is skew Hermitian, and thus has purely imaginary eigenvalues in the Hilbert space $\ch_A$, it is again necessary that $\kee{r}$ is an eigenfunction outside $\ch_A$.

\begin{corollary}\label{cor-similarity}
The moment conditions \eqref{moms} are invariant under any similarity transformation. Namely, for any invertible map $S$, $\brr{l}S^{-1}$, $SFS^{-1}$ and $S\kee{r}$ also fulfill the moment conditions, and thus lead to exact dilations as well.    By choosing $S$ to be a unitary operator, $SFS^{-1}$ remains skew Hermitian.
\end{corollary}
In addition to the variations of the exact dilation scheme with similarity transformations, we will show five examples below, and two more examples in the next section.

\begin{example}[Schr\"odingerization via warped–phase transformation \cite{PhysRevLett.133.230602}]
The Schr\"odingerisation approach can be associated with an ancillary  Hilbert space.
\[
   \ch_A  =  \bigl\{f, f' \in L^{2}((0,\infty),dp)\ \bigm|\ f(0)=0 \bigr\},
   \qquad
\]
with inner product $\langle f,g\rangle:=\int_{0}^{\infty}f(p)^* g(p)  dp.$
In addition, the triple is selected as follows, 
\begin{equation}
    F := -\tfrac{d}{dp}, \quad \kee{r}=e^{-p}, \quad \brr{l} f(p) \mapsto e^{p^\ast}f(p^\ast).
\end{equation}
Here $p^\ast>0 $ is arbitrary.  Another choice of the evaluation is $\brr{l} f(p) \mapsto \int_0^{+\infty} f(p) dp$ \cite{jin2024quantum}.
It is clear that $\kee{r} \notin \ch_A$, due to the violation of the boundary condition. In this case, the function space  $\mathbb X$ can be chosen to be $C_0(0,\infty)$ to contain $\kee{r}$.  In addition, simple calculations verify that the exact moment conditions $\brr{l}F^{k}\kee{r}=1, k=0,1,2,\dots$ are fulfilled. 

\end{example}

\begin{example}[Linear combination Hamiltonian simulations (LCHS) {\cite{ALL23}}]
LCHS expresses the solution of \cref{eq:phys-ode} as an integral with the Lorentzian kernel. To put it into the framework of  \eqref{eq: schrodingerization0}, we can choose the ancillary  Hilbert space to be the real Schwartz class 
endowed with the usual $L^{2}$ inner product,  
\[
   \ch_{A} = \mathcal S(\mathbb R),
   \qquad     
   \langle f,g\rangle=\int_{-\infty}^{\infty}   
       {f(k)^*}  g(k)  dk .
\]
Multiplication by~$ik$ leaves $\ch_A$ invariant, hence we may take the
skew-Hermitian ancillary generator  
\(F  =  ik. \) In addition, we set
\[
   \kee{r}(k)=\frac{-i}{k-i},
   \qquad
   \brr{l} f  = 
   \frac{i}{\pi }\int_{-\infty}^{\infty} \frac{f(k)}{k+i}  dk,
\]
thus giving the LCHS formula,
 \[
\ct e^{ \int_0^t L(s) ds } = \frac{1}{\pi } \int_{\mathbb R} \ct e^{-i \int_0^t H(s) + k K(s)  ds   }  \frac{1}{k^2+1} dk. 
 \]
Although $\kee{r}\notin\mathcal S$,  it still decays algebraically and it can be contained in the tempered-distribution space  $ \mathbb X := \mathcal S'(\mathbb R).$
The moment condition \cref{moms} is then interpreted via contour integrals 
in the complex $k$–plane (with contributions from the simple poles at $k= i$) 
and can be verified using Cauchy’s residue theorem. 
LCHS has been improved to near-optimal query
complexity \cite{ACL23,pocrnic2025constant,huang2025fourier,low2025optimal} and extended to infinite dimensions \cite{lu2025infinite}.
\end{example}

The LCHS has an interesting connection to the Schr\"odingerization method in the previous example.  Let \(\mathcal F\) denote Fourier transform, one can apply the inverse Fourier transform to the triple, and restrict to the half–line \(p>0\): $\bigl(\mathcal F^{-1}\kee{r}\bigr)(p)  = e^{-p}1_{(0,\infty)}(p). $ Similarly, one can show that $\mathcal F^{-1} \circ F \circ \mathcal F= - \frac{d}{dp}, $ and $\brr{l}\mathcal Ff = e^{p^{\ast}}  
         f(p^{\ast})$. 
Thus in the dual \(p\)-representation the triple coincides with the
warped-phase Schrödingerization example.

\medskip 

 Another natural choice for constructing a Hilbert space is by using integral operators. 
 
\begin{example}[Integral–kernel dilation]\label{ex:int}
For a scale $\theta>0$ define the odd kernel
\(
   K_\theta(p)=p  e^{-\theta|p|}.
\)
Its convolution
\( (F_\theta f)(p)=\int_{\mathbb R}K_\theta(p-q)f(q)  dq \)
is skew-Hermitian on $\ch_A=L^{2}(\mathbb R)$.
Meanwhile, for $a\in(-\theta,0)$ one finds  
\(
   F_\theta e^{ax}=\lambda_\theta(a)e^{ax}, 
   \lambda_\theta(a)=-4\theta a/(\theta^{2}-a^{2})^{2}.
\)
Choosing $a$ so that $\lambda_\theta(a)=1$ (e.g.\ $\theta=1\Rightarrow a\approx-0.2253$) gives  
\[
   \ket{r}(x)=e^{ax},\qquad   
   F_\theta\kee{r}=\kee{r}.
\]
By taking $\bra{l}f=f(0)$,   the moment
identities are exactly satisfied.  

Although $\kee{r}\notin\ch_A$, it lies in the weighted space  
\(
   \mathbb X=L^{2}   \bigl(e^{-2\gamma|x|}, dx\bigr)
\)
for any $\gamma\in(-a,\theta)$, where  $\ch_A$ is dense.  Discretising $F$ via Nyström
quadrature yields a narrow-band matrix that might be amenable to 
block-encoding.
\end{example}

Another convenient way to generate skew operators is through
{pseudodifferential} symbols \(a(x,\xi)\):
\[
   \bigl(\Op(a)f\bigr)(x)
   =\frac{1}{2\pi}   
     \int_{\mathbb R}   \int_{\mathbb R}
     e^{  i(x-y)\xi}  
     a(x,\xi)  f(y)  dy  d\xi .
\]

\begin{example}[Dilation via a pseudodifferential generator]\label{ex-ps}
Let the ancillary  Hilbert space be $\ch_A=L^{2}(\mathbb R,dx)$.
Fix an order $0<\theta\le 1$ and define the pseudodifferential generator
\[
   F:=  -  i  (-\Delta)^{\theta},
   \qquad
   \bigl[\widehat{(-\Delta)^{\theta}f}  \bigr](\xi)=\left(\xi^2\right)^{\theta}  \hat f(\xi).
\]
Because $(-\Delta)^{\theta}$ is self-adjoint and positive, the factor
$-i$ makes $F$ skew-Hermitian on~$\ch_A$.
A natural choice for  $\kee{r}$ is the plane wave 
\(
   \kee{r}(x)=e^{i\xi_{0}x}, 
\)
with a complex frequency $\xi_0$ so that $\kee{r} \notin \ch_A$
and the evaluation functional
\(
   \bra{l}f = f(0).
\)
By choosing $\xi_0$ such that $\xi_0^{2\theta}=i$, we have
\(
   F \kee{r}= -i\xi_{0}^{2\theta} \kee{r} = \kee{r}.
\)
Thus all the moment conditions in \cref{moms} are satisfied. 
 Selecting $\arg(\xi_{0}^{2})= \pi/2$ then forces $(\xi_{0}^{2})^{\theta}=e^{  i\pi/2}=i$, i.e.,
\(\xi_{0}=e^{i\pi/(4\theta)}\) from the principal banch.

\end{example}

\begin{example}[Dilations using difference operators]
 Consider a class of difference operators, applied to the space of infinite sequences $\ch_A= \ell^2$, with each $f\in \ell^2$ being $f= (0, f_1, \cdots )$ equipped with the inner product: 
 $\langle f, g\rangle = \sum_{n \geq 1} f_n^* g_n.$ We can define a difference operator $\theta F$ as, $(\theta F f)_0=0$ and  $(\theta F f)_n= \theta( f_{n} - f_{n-1}) $ when $n\geq 1.$
 A telescopic (summation-by-parts) calculation gives
\(F_\theta^\dagger=-F_\theta\) on \(\mathcal H_A\).   In addition, one can verify that with $\kee{r}=\{\lambda_\theta^{  n}\}_{n\ge0},
   \lambda_\theta=\tfrac{\theta}{1+\theta}\in(0,1), $ and $\brr{l} f= f_0$ or $\brr{l} f= \lambda_\theta^{-1} f_1 $ the moment conditions \eqref{moms} are satisfied. 
\end{example}

These examples show the overwhelming possibilities of achieving the exact dilations, in addition to the dilation scheme created by a similarity transformation among $\brr{l}$, $F$, and $\kee{r}$. This wide variety of such dilation methods forms a large platform where different quantum algorithms can be invoked. 
For example,  the dilated Hamiltonian in \cref{ex-ps} and \cref{ex:int}  might be efficiently implemented using the block-encoding techniques in   \cite{li2023efficient,nguyen2022block}.

\smallskip 
In the next two sections, we will closely examine two approaches, one based on a first-order differential equation, for which a finite-difference approximation leads to simple implementation on digital quantum computers,  the other based on Bargman representation of a boson environment \cite{hall2013quantum}, which is more suited for continuous variable quantum platforms.

%====================================================================
\section{A dilation using a function space on a compact domain}
\label{sec:compact-dilation}

The five examples in the previous section work with an ancilla Hilbert space in $\mathbb R$, $\mathbb R_{+}$ or $\mathbb N$, forcing one to truncate tails in a numerical implementation, thus introducing additional complications.  Here we present a model whose dilation acts on a \emph{finite} interval, so the dilation operator can be made uniformly banded for which the algorithmic complexity is easy to control.

%------------------------------------------------------------------
\subsection{Numerical discretizations using SBP}
To find a dilation in \cref{eq: schrodingerization0}, we let
\[
  \mathcal H_A
  := H_0^{1}(0,1)
  = \bigl\{f\in L^{2}(0,1)\ \bigm|\ f'\in L^{2}(0,1),\ f(1)=0 \bigr\},
  \quad
  \langle f,g\rangle = \int_{0}^{1} f(p)^{*}g(p)  dp.
\]
Define the generator as the differential operator $F_\theta= \theta F$, and
\begin{equation}\label{eq:F-cont}
  F := p \partial_{p}+\tfrac12 .
\end{equation}

\begin{lemma}\label{lem:skew-H1}
  For every $\theta>0$ the operator $\theta F$ is skew-Hermitian on
  $\mathcal H_A$.  Moreover, the moment conditions
  of~\eqref{moms} are satisfied with
  \[
     \kee{r} = p^{\beta},
     \qquad
     \brr{l}  f = 2^{\beta} f \left(\tfrac12\right),
  \]
  where the exponent $\beta$ is
  \begin{equation}\label{eq:fromtheta2beta}
     \beta = \tfrac1\theta-\tfrac12.
  \end{equation}
\end{lemma}

\begin{proof}
    Skewness: for $f,g\in\mathcal H_A$, we examine the inner product denoted by $\langle, \rangle$, 
    \begin{equation}\label{eq:ibp}
      \langle f, Fg\rangle
      = \int_0^1 f^*(p g' + \tfrac12 g) dp
      =  p f^* g\Bigr|_0^1 - \langle Ff, g\rangle,
    \end{equation}
    and the boundary term vanishes since $p=0$ at the left endpoint and $f(1)=g(1)=0$ at $p=1$. Hence $F$ is skew-Hermitian and so is $\theta F$.
    For the moments, $F(p^\beta)=(\beta+\tfrac12)p^\beta$; with \eqref{eq:fromtheta2beta} we get
    $F_\theta\ket{r}=\theta(\beta+\tfrac12)\ket{r}=\ket{r}$. Also,
    $\bra{l}\ket{r}=2^{\beta} r(1/2)=2^{\beta}\cdot 2^{-\beta}=1$.
\end{proof}

Note $\ket{r} \notin H_0^{1}$ only because it fails the condition
$r(1)=0$, but $r\in C^{1}(0,1)$ and $r\in H^1(0,1)$ since $\beta>1/2$ when $\theta\in(0,1)$, so the dilation remains
well-posed.

\medskip

%------------------------------------------------------------------
For a qubit implementation, we discretize $\{F,\bra{l},\ket{r}\}$ as follows.
Partition $[0,1]$ into $M$ intervals with grid points $\{p_i\}_{i=0}^M$ and set $v_i=v(p_i)$. The integration by parts that ensured the skew Hermitian property can be extended to the discrete level using summation by parts (SBP). Following \cite{Strand1994,MattssonNordstrom2004}, 
we let $h_j:=p_{j+1}-p_j$ and define the SBP trapezoid weights
\begin{equation}\label{sbp-wh}
W=\mathrm{diag}(w_0,\dots,w_M),\qquad
w_{0}=\tfrac12 h_0,\quad
w_{j}=\tfrac12(h_{j-1}+h_j)\ (1\le j\le M-1),\quad
w_{M}=\tfrac12 h_{M-1}.
\end{equation}
Let $Q$ be the tridiagonal matrix with
\begin{equation}\label{sbp-Q}
    (Q)_{j,j+1}=\tfrac12,\quad (Q)_{j+1,j}=-\tfrac12,\quad
Q_{00}=-\tfrac12,\quad Q_{MM}=+\tfrac12,
\end{equation}
and set $D:=W^{-1}Q$. Then the diagonal-norm SBP identity holds:
\[
WD+D^\dagger W
 = 
Q+Q^\dagger
 = 
B:=\mathrm{diag}(-1,0,\ldots,0,1).
\]
This implies that for all grid functions $\bm u, \bm v$,
$\ \langle \bm u,D \bm v\rangle_W+\langle D\bm u,\bm v\rangle_W = u_M v_M - u_0 v_0$ with weighted inner product $\langle \bm x,\bm y\rangle_W:=\bm x^\dag W \bm y$, thus mimicking the integration by parts property in \cref{eq:ibp}, and thus automatically maintain the skew property after the discretization. Specifically, let $P:=\mathrm{diag}(p_0,\ldots,p_M)$ and define the (Hamiltonian) split form
\begin{equation}\label{sbp-Fh}
    F_w  :=  \tfrac12(PD+DP) - \tfrac12 W^{-1}BP,
\qquad
F_h  :=  W^{1/2}F_w  W^{-1/2}.
\end{equation}
The SBP property automatically guarantees that $F_h$ is skew-Hermitian and tridiagonal, while $F_w$ is skew with respect to the $W$–inner product. On a uniform grid one has $h_i\equiv h:= 1/M$, $p_i=ih$, $i=0,\dots,M$, and $w_0=w_M=h/2$, $w_i =h, \forall 1\leq i \leq M-1 $.  The second-order SBP approximation of $F$ in \cref{eq:F-cont}  reads
\begin{equation}\label{eq:Fh}
  \begin{aligned}
    (F_h v)_0 &= \tfrac{1}{2\sqrt2}  v_{1},\\
    (F_h v)_1 &= \tfrac34  v_{2}-\tfrac1{2\sqrt2}  v_{0},\\
    (F_h v)_i &= \tfrac{p_{i+1}+p_i}{4h}  v_{i+1}
              - \tfrac{p_i+p_{i-1}}{4h}  v_{i-1},
              \quad 2\le i\le M-1,\\
    (F_h v)_M &= -\frac{p_{M-1}+p_M}{4h}  v_{M-1}.
  \end{aligned}
\end{equation}
The skew-Hermitian property is apparent.

With the approximation of $F$ by $F_h \in \mathbb{C}^{(M+1) \times (M+1)}$, we obtain a dilated Hamiltonian,
\begin{equation}\label{dilated-H}
    \widetilde H(t) :=  I\otimes H(t) +  i \theta F_h\otimes K(t).
\end{equation}

The implementation of $ \widetilde H(t)$ requires an appropriate representation of $F_h$ on quantum devices.  The SBP split form can be written as $F_h=a-a^\dagger$, where
 \(a:=\sum_{j=0}^{M-1} f_j  \ketbra{j}{j+1}\) 
 is a weighted lowering operator. In this sense \(iF_h\) plays the role of a discrete momentum/current operator on a tight-binding
 chain (a discrete derivative up to constants). With \(|r_h\rangle=W^{1/2}\mathbf 1\), the right vector is the mass-weighted flat  mode: for \(\theta=2\), the SBP closure gives
 \(\theta F_h|r_h\rangle=|r_h\rangle\) on all interior sites.

We can also associate $F_h$ with Pauli strings, which provide more alternatives to the implementation on digital quantum computers.
Using a unary (one–hot) encoding of the ancilla with computational basis \({\ket{j}}_{j=0}^{M}\), define \(\sigma_j^\pm=(X_j\mp iY_j)/2\). For the off-diagonal couplings \(f_j\), the split–form chain satisfies
\(
(F_h)_{j,j+1}=f_j,\qquad (F_h)_{j+1,j}=-f_j,
\)
hence
$
(iF_h)_{j,j+1}=i f_j,\qquad (iF_h)_{j+1,j}=-i f_j.
$
It admits the 2–local Pauli expansion
\[
F_h
=\sum_{j=0}^{M-1} f_j\bigl(\sigma_j^+\sigma_{j+1}^- -  \sigma_j^-\sigma_{j+1}^+\bigr),
\Longrightarrow
iF_h
= -\tfrac12\sum_{j=0}^{M-1} f_j\bigl(X_j Y_{j+1}-Y_j X_{j+1}\bigr).
\]

Meanwhile, for the initial state $\kee{r}$ (an eigenvector of $F$), we take the direct discretization with normalizing factor $Z_\beta$,
\begin{equation}\label{r_h}
   \kee{r_h}= \ket{r_h} = Z_{\beta}^{-1}
                \sum_{j=0}^{M} p_j^{\beta}  w_j^{1/2}  \ket{j}, 
    \qquad 
    Z_\beta = \Bigl(\sum_{j=0}^{M} w_j p_j^{2\beta}\Bigr)^{1/2} 
            = \co  \left( \frac{1}{\sqrt{2\beta+1}}  \right).
\end{equation}
Take
\(
  \brr{l_h} = \frac{1}{\braket{j_\ast}{r_h}}  \bra{j_\ast}, 
\)for some $0\leq j_\ast< M $.
Then $\big( l_h|r_h\big)=1$ by construction.

These elements constitute a simple procedure to simulate the ODEs \eqref{eq:phys-ode}: Starting with $\ket{r_h} \otimes \ket{\bm x_0}$, one applies Hamiltonian simulations with the dilated Hamiltonian \eqref{dilated-H} to time $T$, and then post-select at $j_\ast$ on the ancilla.  In the following sections, we will analyze the error due to the discretization and the success probability from the post-selection.

\medskip

%We first discuss the case when the grid is uniform $h_i=h=1/M$.
%The inclusion of the quadrature weights $w_j$ in $\ket{r_h}$ ensures that $\bra{l_h} (\theta F_h)\ket{r_h}$ mimics the integral. 

\subsection{Error Estimation using a Light-cone Property}

Considering the simple choice $\beta=0$, i.e., $\theta=2$, then 
$\ket{r_h} = \sqrt{h/2} \ket{0} +\sqrt{h} \ket{1} + \cdots + \sqrt{h} \ket{M-1}  + \sqrt{h/2} \ket{M}$.
A direct calculation shows that
\(
\bra{j}  \theta F_h \ket{r_h}=\bra{j}\ket{r_h}
\) for all $j<M$, which is guaranteed by the SBP construction. Similarly, \(
\bra{j}  (\theta F_h)^2 \ket{r_h}=\bra{j}\ket{r_h}
\) for all $j<M-1$. 
Therefore,
\begin{equation}
    \brr{l_h} (\theta F_h)^k \kee{r_h} =1, \; \forall  k<M-j_\ast,
\end{equation}
which is an indication of the approximate property via a finite moment matching of \eqref{moms}.

\medskip

\medskip

One remarkable observation is that  by adding one diagonal entry to $F_h$,
\begin{equation}\label{alpha}
  \widehat{F}_h = F_h + \alpha \ketbra{M}{M}, \quad   \alpha
 = 
\frac{1}{\theta}
 - 
\frac{\bra{M}F_h\ket{r_h}}{\bra{M}\ket{r_h}} = \co(M).
\end{equation}
we obtain $\theta \widehat{F}_h\ket{r_h} = \ket{r_h}$, implying that the triple
\(
\bra{l_h}, \widehat{F}_h, \ket{r_h}
\)
satisfies moment conditions of \emph{all orders} with a finite-dimensional, almost skew-Hermitian operator $\widehat{F}_h$.
In light of \cref{thm:moment-recovery}, the dilation using a finite-dimensional non-Hermitian $\widehat{F}_h$ recovers the exact solution of the ODEs \eqref{eq:phys-ode}. We will refer to such correction as moment-locking closure (MLC). 

The benefit of MLC is two-fold. First, using $\widehat{F}_h$ provides a simple route to quantify the error by the dilated Hamiltonian without going to infinite dimensions. Although $\widehat{F}_h$ is non-Hermitian, it provides an anlysis device to quantify the error. Secondly, compared to the infinite-dimensional formulae, such as  Schrödingerization \cite{PhysRevLett.133.230602}, LCHS \cite{ALL23}, and \eqref{eq: schrodingerization0}, it provides an easier starting point to design successively improved approximations.   

To elaborate on the first point, we also define $\widehat H(t) :=  I\otimes H(t) +  i \theta \widehat F_h\otimes K(t)$, along with the dynamics they generate, 
\begin{align}
    i \frac{d}{dt}\Psi(t)& = \widetilde H(t) \Psi(t), \quad \Psi(0) = \ket{r_h} \otimes \bm x_0, \\
    i \frac{d}{dt}\Phi(t)& = \widehat H(t) \Phi(t), \quad \Phi(0) = \ket{r_h} \otimes \bm x_0.
\end{align}
Since  $\theta \widehat F_h$ has $\ket{r_h}$ as an eigenvector with eigenvalue 1, $\Phi(t)= \ket{r_h} \otimes \bm x(t)$.
Define the error $\chi(t):=\Phi(t)-\Psi(t)$. Subtracting the two evolutions yields
\[
i \frac{d}{dt}\chi(t) = \widetilde H(t) \chi(t) + i\theta \alpha \big(\ketbra{M}{M}\otimes K(t)\big) \Phi(t),
\quad \chi(0)=0.
\]
In the error equation, the source of the error clearly concentrates at the right boundary only in the ancilla. Since $F_h$ in $\widetilde{H}$ only has nonzero entries at off-diagonals $(j,j+1)$ and $(j+1,j)$, the error propagation to the left, especially the arrival at $j_\ast$, where we post-select the solution, can be estimated by a light-cone type of approach \cite{Lieb1972}, which in the context of the differential operator in \cref{eq:F-cont}, is known as the CFL condition \cite{leveque1992numerical}.     
  \cref{fig:cfl} shows a numerical demonstration of such back propagation of the boundary error using the ODE $\frac{d}{dt}\bm v= -\theta \kappa(t) F_h  \bm v$, where one sees the boundary perturbation confined to the characteristic cone emanating from $p=1$ with speed $-\theta p$; points outside that domain of dependence remain unchanged. The arrival time at a location $p_\ast$ scales like
  \begin{equation}\label{Tarrival}
      T(1 \to p_\ast) = -\ln p_\ast.
  \end{equation}
  See \cref{proof-thm1} and \cref{lem:fps-sharp} for a derivation. We now quantify such error propagation. 

\begin{figure}[htp]
\centering
\includegraphics[scale=0.5]{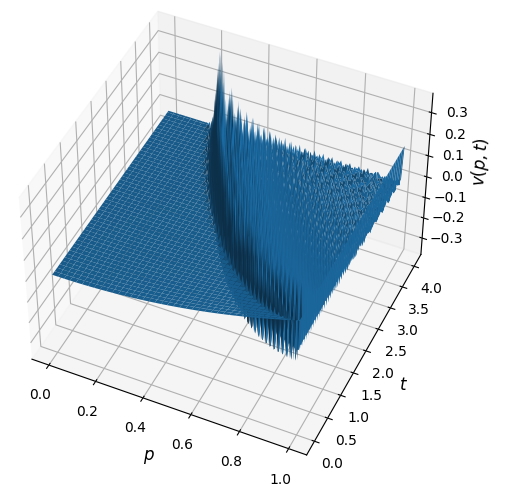}
\caption{Snapshots of the semi-discrete solution for
$\bm v' = -\theta F_h \bm v$ with
$v_j(0) \propto p_j^{1/\theta - 1/2}$.}
\label{fig:cfl}
\end{figure}

\begin{lemma}[Finite–speed propagation bound]\label{lem:cfl-chi}
Let $\theta=2$ and consider the uniform grid with grid size $h=1/M$.  Suppose $m=M-j_\ast$ and that 
\begin{equation}\label{a-CFL}
    \varrho \;:=\; \frac{2 e \theta K_{\max} T}{mh}<\frac12,
\end{equation}
then the boundary–induced error arriving at $p_\ast=j_\ast h$ obeys
\begin{equation}\label{eq:chi-final-bound-main}
\bigl\|(\bra{j_\ast} \otimes I) \chi(T)\bigr\|
 \;\le\;
2 \alpha \bigl|\braket{M}{r_h}\bigr|\;X(T) K_{\max}T\;\Bigl(\dfrac12\Bigr)^{ M - j_\ast},
\qquad
X(T):=\sup_{t\in[0,T]}\|\bm x(t)\|.
\end{equation}
Consequently, to guarantee $\bigl\|(\bra{j_\ast} \otimes I) \chi(T)\bigr\|\le\epsilon$, it suffices to choose
\begin{equation}\label{eq:two-reqs}
M - j_\ast \ge\; 
 \log_2 \frac{4 \alpha X(T) K_{\max}T}{\epsilon}.
\end{equation}
\end{lemma}

As a result, under the condition in \cref{a-CFL}, especially when $K_{\max}T = \co(1)$ we can make the geometric factor
smaller than $\epsilon$ by taking $M=\Omega \big(\log(K_{\max}T X(T)/\epsilon)\big)$ (Thus the number of ancilla has $\log\log(K_{\max}T X(T)/\epsilon)$ scaling). 
Since $F_h$ is $2$–sparse and $\|F_h\|_2=\Theta(1/h)=\Theta(M)$, by using the matrix norm bounds for sparse matrices,
the dilated Hamiltonian \eqref{dilated-H} obeys
\[
\norm{\widetilde{H}(t)}
 = 
\co \big(H_{\max} + \theta K_{\max} M\big)
 = 
\co \Big(H_{\max} + \theta K_{\max} \log \frac{K_{\max}T X(T)}{\epsilon}\Big),
\]
where $H_{\max}:=\max_{0\le t\le T}\|H(t)\|$.  Leveraging the time–dependent Hamiltonian simulations using Dyson series
\cite{low2018hamiltonian},  we arrive at,
\begin{theorem}[ODE simulations by dilation under the light-cone condition \eqref{a-CFL}]\label{thm:cfl-main} Assume that $K_{\max}T=\co(1),$ and
  assume block encodings of $H(t)$ and $K(t)$ with norms
  $H_{\max} $ and $  K_{\max}$, respectively. The dilation scheme outputs an $\epsilon$-accurate approximate state for $\ket{\bm x(t)}$ with
  success probability $\Omega(1)$ using
  \[
    \widetilde{\mathcal O}   \left(
     \frac{  \big(T H_{\max}+\log \frac{X(T)}{\epsilon} \big) \norm{\bm x(0)} }{\norm{\bm x(T)}}
    \right)
  \]
  block-encoding queries to $H(t)$ and $K(t)$, and a comparable number of gates,  where $\widetilde{\mathcal O}$ hides the $\frac{\log T H_{\max}/\epsilon }{ \log\log T H_{\max}/\epsilon}$ factor.   For autonomous ODEs where $K(t)$ and $H(t)$ are time-independent, the query complexity is reduced to,
   \[
    \co   \left(
       \big(T H_{\max}  +  \log \frac{X(T)}{\epsilon} + \log \frac{T H_{\max} }{\epsilon} \big) \frac{\norm{\bm x(0)}}{\norm{\bm x(T)}} \right).
     \] 
\end{theorem}

\subsection{Geometrically refined grids with an improved light-cone property}

The arrival-time estimate based on the PDEs generated by $F$ suggests that the propagation of the boundary effect slows down when it is approaching the left boundary $-\log p_\ast$. This suggests that we mitigate a large $K_{\max}T$ by refining only near the left boundary while
keeping the bulk step size moderate. We therefore adopt a \emph{geometric} grid
\begin{equation}\label{p-geom}
  p_j \;=\; \exp \Big[-\delta\Big(M-{j}\Big)\Big] ,\qquad j=0,1,\dots,M,
\end{equation}
with grading parameter $\delta>0$ ($\delta=1$ is enough for the discussions below). One can directly verify that the
(trapezoidal) diagonal quadrature weights are
\[
w_j=\begin{cases}
\dfrac{1}{2} (e^{\delta}-1) p_0, & j=0,\\[6pt]
\dfrac{p_{j+1}-p_{j-1}}{2}\;=\;p_j \sinh\delta, & 1\le j\le M-1,\\[6pt]
\dfrac{1}{2} (1-e^{-\delta}) p_M, & j=M~,
\end{cases}
\]
and the split–form skew tridiagonal off–diagonals
$f_j:=(F_h)_{j,j+1}$ in \cref{sbp-Fh} are uniform in the interior, which makes it easier to analyze and implement. Specifically, we have, 
\begin{equation}\label{eq:geom-const-w}
f_j=\frac{1}{4\sinh(\delta/2)}\times
\begin{cases}
\sqrt{ 1+e^{-\delta} }, & j=0,\\[4pt]
1, & 1\le j\le M-2,\\[4pt]
\sqrt{ 1+e^{\delta} }, & j=M-1~.
\end{cases}
\end{equation}
We consider $\delta \geq \delta_0$ with a constant $\delta_0$,  which implies that,
\begin{equation}\label{eq:geom-norm}
  \norm{F_h} \;\le\; 2 \max_{0\le j\le M-1} f_j
  \;=\; \frac{\sqrt{ 1+e^{\delta} }}{2 \sinh(\delta/2)} = \Theta(1).
\end{equation}

We now state the finite–propagation estimate on the geometric chain. The proof is shown in \cref{sec: CFL-geom}.

\begin{lemma}[Finite-speed propagation on a geometric chain]\label{lem:geom-CFL}
Let $F_h$ be the skew–Hermitian chain with  off–diagonals $f$ in 
\eqref{eq:geom-const-w} and $j_\ast \in [M-1]$. 
If
\begin{equation}\label{cfl-geom}
    \varrho:= \frac{e\theta K_{\max}T}{4 (M-j_\ast) \sinh(\tfrac{\delta}2)} <1,  
\end{equation} then the boundary influence decays geometrically. 
\begin{equation}\label{eq:geom-CFL-decay}
  \norm{\langle j_\ast| \otimes I  \chi(T) }
  \;\le\; \frac{\varrho^{ M-j_\ast}}{1-\varrho^2}.
\end{equation}
\end{lemma}

\noindent
The condition in \eqref{cfl-geom} is significantly better than \cref{a-CFL}. 
Choosing $\delta = \Theta(1)$ and $M=\Theta(K_{\max}T)$ is sufficient
to make the base $\varrho<1$, and an additional $\co(\log(1/\epsilon))$ points pushes the error below~$\epsilon$. Furthermore, one can choose a much larger $\delta$ to enhance this property, e.g., $\delta= \log(K_{\max} T).$ However, the downside of this approach is that the right vector $\braket{j}{r_h}\propto W^{1/2}_j $ has exponentially small components near the left boundary, thus leading to a small success probability. On the other hand, by choosing a large post-select window $[0, p_{j_\ast}]$, the small probability can be improved to $\Omega(1)$, but in this case, \cref{cfl-geom} is difficult to satisfy.

So for long-time simulation where $K_{\max} T \gg1$, we divide the simulation into segments, each with length $\tau \co(1/K_{\max})$. Upon the completion of the Hamiltonian evolution of each segment, one obtains an accurate solution in the window $[0, p_{j_\ast}]$, thanks to the light-cone condition in \cref{cfl-geom}. To transition to the next time segment, we simply apply oblivious amplitude amplifications (OAA) to restore the ancilla state back to $\ket{r_h}$.  

\begin{corollary}[Light-cone protected window with constant weight]\label{cor:geom-CFL-window}
Fix constants $\delta_0>0$ and $\Delta\in(0,1/2)$, and assume $\delta\ge\delta_0$.
Let $m:=M-j_\ast\ge 1$ and choose
\begin{equation}\label{eq:m-choice}
m \;\ge\; \max\!\left\{
\frac{1}{2\sinh(\delta/2)}\,,\;
\frac{1}{\delta}\log\!\left(\frac{1+e^\delta}{2\Delta}\right)
\right\}.
\end{equation}
If the segment length $\tau$ satisfies
\begin{equation}\label{eq:tau-choice}
\tau \;\le\; \frac{1}{e\,\theta\,K_{\max}},
\end{equation}
then for the readout site $j_\ast:=M-m$ we have:

(i) (Light-cone) the geometric-chain parameter obeys
\[
\varrho=\frac{e\theta K_{\max}\tau}{4m\sinh(\delta/2)} \;\le\; \frac12,
\]
so Lemma~\ref{lem:geom-CFL} applies at $j_\ast$ over time~$\tau$;

(ii) the window probability satisfies
\[
P_{\mathrm{win}}(j_\ast)
=\frac{\frac12(p_{j_\ast}+p_{j_\ast+1})-p_0}{1-p_0}
\;\ge\; \Delta - p_0,
\qquad p_0=e^{-\delta M}.
\]
Since $p_0$ is exponentially small with $M$, we can pick a safe lower bound $\Delta/2$ by choosing a sufficiently large $M$.

(iii) $\norm{F_h} = \co(1).$

\end{corollary}

\begin{proof}
(i) Using \eqref{eq:tau-choice} gives
\[
\varrho=\frac{e\theta K_{\max}\tau}{4m\sinh(\delta/2)}
\le \frac{1}{4m\sinh(\delta/2)}\le \frac12,
\]
by the first lower bound in \eqref{eq:m-choice}.

(ii) With $j_\ast=M-m$, we have $p_{j_\ast}=e^{-\delta m}$ and
$p_{j_\ast+1}=e^{-\delta(m-1)}$, hence
\[
P_{\mathrm{win}}(j_\ast)
=\frac{\frac12\bigl(e^{-\delta m}+e^{-\delta(m-1)}\bigr)-p_0}{1-p_0}
\ge \frac12\bigl(e^{-\delta m}+e^{-\delta(m-1)}\bigr)-p_0
= \frac{1+e^\delta}{2}e^{-\delta m}-p_0.
\]
The second lower bound in \eqref{eq:m-choice} ensures
$\frac{1+e^\delta}{2}e^{-\delta m}\ge \Delta$, giving the claim. 

(iii) This comes directly from the magnitude of $f_j$'s. 
\end{proof}

\subsection{Segmentation-wise simulation and oblivious amplitude amplification}

We describe a segmented algorithm for approximating the normalized solution.
Specifically, consider a partition $ [0,T] = \bigcup_k [t_k,t_{k+1}]$, where $t_k:=k\tau$, with $\tau$ selected as follows
\begin{equation}\label{seg-tau}
    \tau K_{\max}=\mathcal O(1). 
\end{equation}
We also let
\begin{equation}
    \mathcal U_k:=\mathcal T\exp\!\bigl(-i\int_{t_k}^{t_{k+1}}\widetilde H(s)\,ds\bigr)
\end{equation}
be the segment propagator on $[t_k,t_{k+1}]$. After each segment, we restore the ancilla back to $\ket{r_h}$ using OAA. The overall workflow is depicted in \cref{fig:circuit_segment}. 
\begin{figure}[h]
\centering
\begin{tikzpicture}[
    scale=1,
    % Define styles for blocks
    operator/.style={draw, fill=blue!10, rectangle, minimum height=2.5cm, minimum width=2.5cm, align=center, rounded corners},
    restore/.style={draw, fill=green!10, rectangle, minimum height=1cm, minimum width=2cm, align=center},
    wire/.style={thick}
]

% --- Coordinates ---
% Rails
\coordinate (sys_in) at (0, 1);
\coordinate (anc_in) at (0, -1);

% --- Step 1: First Segment ---
% Input Labels
\node[left] at (sys_in) {System $|x(0)\rangle$};
\node[left] at (anc_in) {Ancilla $|r_h\rangle$};

% Unitary Block 1
\node[operator] (U1) at (3, 0) {\Large $U \!= \!\ct \!e^{-i \int \widetilde{H}ds }$};

% Connections to U1
\draw[wire] (sys_in) -- (U1.west |- sys_in);
\draw[wire] (anc_in) -- (U1.west |- anc_in);

% --- Step 2: OAA Restoration ---
% System passes through essentially unchanged (logically)
% Ancilla goes into OAA
\node[restore] (OAA) at (6.5, -1) {OAA\\Restore};

% Connections U1 -> OAA
\draw[wire] (U1.east |- sys_in) -- (9, 1); % System rail continues
\draw[wire] (U1.east |- anc_in) -- (OAA.west); % Ancilla rail to OAA

% --- Step 3: Second Segment ---
% Unitary Block 2
\node[operator] (U2) at (10.5, 0) {\Large $U \!= \ct \! e^{-i \int  \widetilde{H}ds } $};

% Connections OAA -> U2
\draw[wire] (9, 1) -- (U2.west |- sys_in); % System enters U2
\draw[wire] (OAA.east) -- (U2.west |- anc_in) node[midway, above, scale=0.8] {$|r_h\rangle$}; % Restored Ancilla enters U2

% Output
\draw[wire] (U2.east |- sys_in) -- (13, 1) node[right] {$|x(2\tau)\rangle$};
\draw[wire] (U2.east |- anc_in) -- (13, -1);

% --- Annotations ---
% Brace for Segment 1
\draw[decorate, decoration={brace, amplitude=10pt, mirror}, thick, gray] (1.5, -2) -- (7.5, -2) node[midway, below=15pt] {Segment 1: Evolution \& Restoration};

% Brace for Segment 2
\draw[decorate, decoration={brace, amplitude=10pt, mirror}, thick, gray] (9, -2) -- (12, -2) node[midway, below=15pt] {Segment 2: Repeat};

\end{tikzpicture}
\caption{Circuit diagram for the segmented simulation. The algorithm proceeds in time steps of duration $\tau$. In each segment, the joint evolution $U_k = \ct \exp(-i \int  \widetilde{H}ds )$ is applied. The system state evolves from $|x(t)\rangle$ to $|x(t+\tau)\rangle$ (conditioned on success), while the ancilla is restored to the reference state $|r_h\rangle$ via Oblivious Amplitude Amplification (OAA) before the next segment begins.}
\label{fig:circuit_segment}
\end{figure}

The light-cone property ensures that the prefront window $\mathrm{win}=\{0,1,\dots,j_\ast\}$ on the ancilla lattice is unaffected by the boundary error (up to an $\epsilon$ error), and retains the solution $  \bm x(t)$ within the $\mathrm{win}$. In particular, within the choice of \cref{seg-tau} for $\tau$, $P_\mathrm{win} : = \sum_{j \in \mathrm{win} } \abs{\braket{j}{r_h}}^2 $ is a constant lower bound thanks to \cref{cor:geom-CFL-window}, and without loss of generality, we may set $P_\mathrm{win} \geq 1/4.$ Let us define the corresponding
projector and the truncated
ancilla mode, respectively,
\begin{equation}\label{eq:PiW}
\Pi_{\mathrm{win}} \;:=\; \Big(\sum_{j\in \mathrm{win} }\ket{j}\bra{j}\Big)\otimes I, \quad \ket{r_{\rm win}}:=\sum_{j\in \mathrm{win} }\ket{j}\bra{j}\ket{r_h}.
\end{equation}

Write the (normalized) solution from the previous segment  $k$ as $\ket{\Psi_k}$.
The post-window state decomposes as
\begin{equation}\label{eq:good-bad-decomp}
\mathcal U_k\ket{\Psi_k}
\;=\;
\underbrace{\Pi_\mathrm{win}\mathcal U_k\ket{\Psi_k}}_{=:~\ket{G_k}}
\;+\;
\underbrace{(I-\Pi_\mathrm{win})\mathcal U_k\ket{\Psi_k}}_{=:~\ket{B_k}},
\qquad
\braket{G_k}{B_k}=0.
\end{equation}
The success probability of a direct window postselection is $p_k:=\|\ket{G_k}\|^2$.

Since the boundary-induced component that
reaches the window $\mathrm{win}$ during a single segment is exponentially small,  there exists an error vector $\ket{\bm e_k}$ supported on $\mathrm{win}$
such that
\begin{equation}\label{eq:window-structure}
\Pi_\mathrm{win} \mathcal U_k\ket{\Psi_k}
\;=\;
\ket{r_{\rm win}}\otimes\ket{\bm x(t_{k+1})}
\;+\;
\ket{\bm e_k},
\qquad
\|\bm e_k\|\le \epsilon.
\end{equation}

Equation \eqref{eq:window-structure}
formalizes an important observation that the unwanted boundary disturbance has \emph{negligible overlap} with
the ``good'' component supported on the window $\mathrm{win}$, and one can invoke oblivious amplitude amplification (OAA).  This requires reflections defined by $\Pi_\mathrm{win}$ and  $\ket{r_h}$, the ancilla reference used at the beginning of
each segment. Define the two reflections
\begin{equation}\label{eq:reflectors}
R_\mathrm{win} := I-2\Pi_\mathrm{win},
\qquad
R_{r} := 2(\ket{r_h}\bra{r_h}\otimes I)-I.
\end{equation}
Both are implementable with ancilla-only controls. The OAA applies 
iterations 
\begin{equation}\label{eq:grover-iterate}
Q_k \;:=\; -\,R_{r}\,\mathcal U_k^\dagger\,R_\mathrm{win}\,\mathcal U_k.
\end{equation}
In the two-dimensional invariant subspace spanned by $\ket{G_k}$ and $\ket{B_k}$,
$Q_k$ acts as a rotation that moves  $\ket{r_{\rm win}} $ to $\ket{r_h}$.

Operationally, in the subspace \(\mathrm{span}{\ket{G_k},\ket{B_k}}), (Q_k)\) performs a Grover rotation that amplifies the \(\Pi_{\mathrm{win}}\) (‘good’) component of \(\mathcal U_k\ket{\Psi_k}\); an additional ancilla-only unitary then restores the amplified window state \(\ket{r_{\rm win}}/|r_{\rm win}|\) to \(\ket{r_h}\).

OAA requires $\Theta\!\Big(\frac{1}{\sqrt{p_k}}\Big)$ rounds to increase $\sqrt{p_k}$ to 1.  In each segment, the raw success probability factors as
$p_k \approx P_{\rm win}(j_\ast)\,\|\bm x(t_{k+1})\|^2/\|\bm x(t_k)\|^2$
(up to the exponentially small light-cone leakage into $\mathrm{win}$), so the
OAA iteration count scales like $1/\sqrt{p_k}$.  By choosing $j_\ast$ (equivalently, the
window) so that $P_{\rm win}(j_\ast)\ge \Delta=\Omega(1)$, we ensure that the window
overlap contributes only a constant amplification overhead. Thus, with a renormalization after each segment,
the OAA overhead per segment scales like
\begin{equation}\label{eq:per-seg-oaa-factor}
R_k \;=\; \mathcal O\!\Big(\max\Big\{1,\frac{\|\bm x(t_k)\|}{\|\bm x(t_{k+1})\|}\Big\}\Big).
\end{equation}

For monotone dissipative dynamics ($\|\bm x(t)\|$ nonincreasing) the cumulative
amplification overhead is
\[
\prod_k \|\bm x(t_k)\|/\|\bm x(t_{k+1})\|=\|\bm x(0)\|/\|\bm x(T)\|.
\]

Another interesting scenario is when the solution goes through transient growth. 
If $\|\bm x(t)\|$ is \emph{not} monotone and exhibits transient amplification before decaying,
then \eqref{eq:per-seg-oaa-factor} shows that OAA is only costly on segments where the norm
\emph{decreases}. One may therefore omit amplification on growth segments (where $p_k$ is larger),
and apply OAA only after the trajectory passes a local maximum. For example, when $\norm{\bm x(t)}$ experiences a transient growth followed by a stabilized path with a ``hump'' profile,  the total
amplification overhead is controlled by the peak-to-final ratio,
\begin{equation}\label{eq:transient-overhead}
\prod_{k}\max\Big\{1,\frac{\|\bm x(t_k)\|}{\|\bm x(t_{k+1})\|}\Big\}
\;=\;
\frac{\sup_{t\in[0,T]}\|\bm x(t)\|}{\|\bm x(T)\|}.
\end{equation}

Define the cumulative
amplification factor
\begin{equation}\label{eq:Gamma-def}
\Gamma
:=\prod_{k=0}^{N-1}\max\Big\{1,\frac{\|\bm x(t_k)\|}{\|\bm x(t_{k+1})\|}\Big\}.
\end{equation}

\begin{theorem}[Query complexity via segmentation + OAA]\label{thm:seg-complexity}
Assume block encodings of $H(t)$ and $K(t)$, choose $\delta\ge\delta_0>0$ so that
$\|F_h\|=\Theta(1)$.  Let $\tau$ satisfy \eqref{seg-tau}, and let $N=\lceil T/\tau\rceil$. Choose $(M,j_\ast)$ so that $P_{\rm win}(j_\ast)\ge\Delta=\Omega(1)$ and the light-cone leakage per segment and Hamiltonian simulation error is $\epsilon_{\rm seg}\le\Theta(\epsilon/N)$.

\smallskip
(i) \emph{Non-autonomous dynamics.} There is a segmented algorithm based on Dyson-series
time-dependent simulation of $\mathcal U_k$ and OAA restoration such that the total number
of block-encoding queries is
\begin{equation}\label{eq:total-td-refined}
Q^{\rm (td)}
=
\mathcal O\!\left(
\big(H_{\max}T + K_{\max}T\big)\;
\frac{\log\!\big(\Lambda_\tau K_{\max} T /\epsilon\big)}
{\log\log\!\big(\Lambda_\tau K_{\max} T /\epsilon\big)} \Gamma
\right),
\qquad
\Lambda_\tau := H_{\max}\tau + O(1).
\end{equation}

where $\Gamma$ is defined in \eqref{eq:Gamma-def}. In particular, if $K(t)\le 0$ for all $t$
then $\Gamma=\|\bm x(0)\|/\|\bm x(T)\|$.

\smallskip
(ii) \emph{Autonomous dynamics.} If $H(t)\equiv H$ and $K(t)\equiv K$, then using qubitization/QSVT
per segment yields
\begin{equation}\label{eq:total-ti-refined}
Q^{\rm (ti)}
=
\mathcal O\!\left( \Big ( 
\big( \norm{H}+ \norm{K}\big)T
\;+\;
 \norm{K} T\log(\norm{K}T/\epsilon) \Big) \Gamma 
\right). 
\end{equation}

%\begin{equation}\label{eq:total-ti}
%Q^{\rm (ti)} = \widetilde{\mathcal O}\!\Big((H_{\max}+K_{\max})\,T\cdot \Gamma \;+\; K_{\max}T\cdot \log(K_{\max}T/\epsilon)\cdot \Gamma\Big).
%\end{equation}
\end{theorem}

\begin{proof}[Proof sketch]
Fix $\epsilon_{\rm seg}=\Theta(\epsilon/N) =\Theta(\epsilon/(K_{\max} T)) $ and apply Hamiltonian simulation algorithms to each segment.
\end{proof}

\paragraph*{Complexity perspective.}
From a resource-scaling viewpoint, the segmented tight-binding dilation inherits the near-optimal precision dependence of interaction-picture simulation: the overhead in $\epsilon$ enters only through the standard
$\log(\cdot)/\log\log(\cdot)$ factor from truncated Dyson-series techniques. In contrast, a direct LCHS-based LCU embeddings express the non-unitary propagator as a spectral integral over Hamiltonian evolutions, and after smoothing and quadrature, this typically introduces an \emph{additional} polylogarithmic overhead in the target accuracy \cite{ACL23}. 
In the autonomous setting, qubitization/QSVT applied segment-wise yields an \emph{additive} logarithmic dependence on $1/\epsilon$ governed by the dissipative scale (the $\log(1/\epsilon)$ term is weighted by $\|K\|T$, not by $\|H\|T$). Therefore, for weakly dissipative and autonomous problems, our complexity is slightly better. 
Finally, in the unitary limit $K\to 0$ the dilation becomes unnecessary and we recover the optimal Hamiltonian-simulation scaling for $e^{-iHT}$ (up to the usual additive $\log(1/\epsilon)$ term).

\subsection{Numerical Illustration: 2-D Maxwell Viscoelasticity model}
\label{sec:maxwell-test}

Here we present a numerical test to demonstrate the performance of the dilation scheme in the previous section. 
We consider wave propagation in a 2D  solid with viscoelastic behavior. The model involves a stress-strain pair $(\sigma, \epsilon)$, augmented by an internal variable $\gamma$. The variables obey Maxwell type  constitutive relations:
$$
\sigma = K_1(\epsilon - \gamma),
\qquad
\eta \dot{\gamma} = \sigma - K_2 \gamma .
$$Here, $K_1$ is the series elastic modulus, $K_2$ is the parallel elastic modulus, and $\eta$ is the viscosity \cite{lakes2009}.  To express the PDEs in the form of \eqref{eq:phys-ode} with the decomposition \eqref{A2HK},
we define a new state vector ${\bm u} = (u_1, {\bm u}_2, u_3)^T$ where:$$
u_1 = \sqrt{K_1}(\epsilon - \gamma), \quad {\bm u}_2 = \sqrt{\rho^{-1}} {\bm p}, \quad u_3 = \sqrt{K_2} \gamma,
$$
and  we have denoted the momentum as $\mathbf p=(p_x,p_y)^{T}$.

The coupled PDEs for the transformed state vector ${\bm u}$ can be expressed in the first-order form 
\begin{equation}
    \frac{\partial}{\partial t}
\begin{pmatrix}
u_1 \\ {\bm u}_2 \\ u_3
\end{pmatrix}
=
\left(
\underbrace{
\begin{pmatrix}
0 & \frac{\sqrt{K_1}}{\sqrt{\rho}}\nabla\cdot & 0 \\
\frac{\sqrt{K_1}}{\sqrt{\rho}}\nabla & {0} & 0 \\
0 & {0} & 0
\end{pmatrix}
}_{-iH }
+
\underbrace{
\begin{pmatrix}
-K_1/\eta & 0 & \sqrt{K_1 K_2}/\eta \\
{0} & {0} & {0} \\
\sqrt{K_1 K_2}/\eta & {0} & -K_2/\eta
\end{pmatrix}
}_{K}
\right)
\begin{pmatrix}
u_1 \\ {\bm u}_2 \\ u_3
\end{pmatrix}
\end{equation}
To  maintain the structure of these operators, we apply a direct finite-difference scheme to approximate the gradient and divergence on a 2D grid in the domain $[0,2]\times [0,2]$ with $64 $ grid points in each direction. The total number of unknowns is thus $N=16384.$

\begin{figure}[thp]
  \centering
  \includegraphics[width=3.37in]{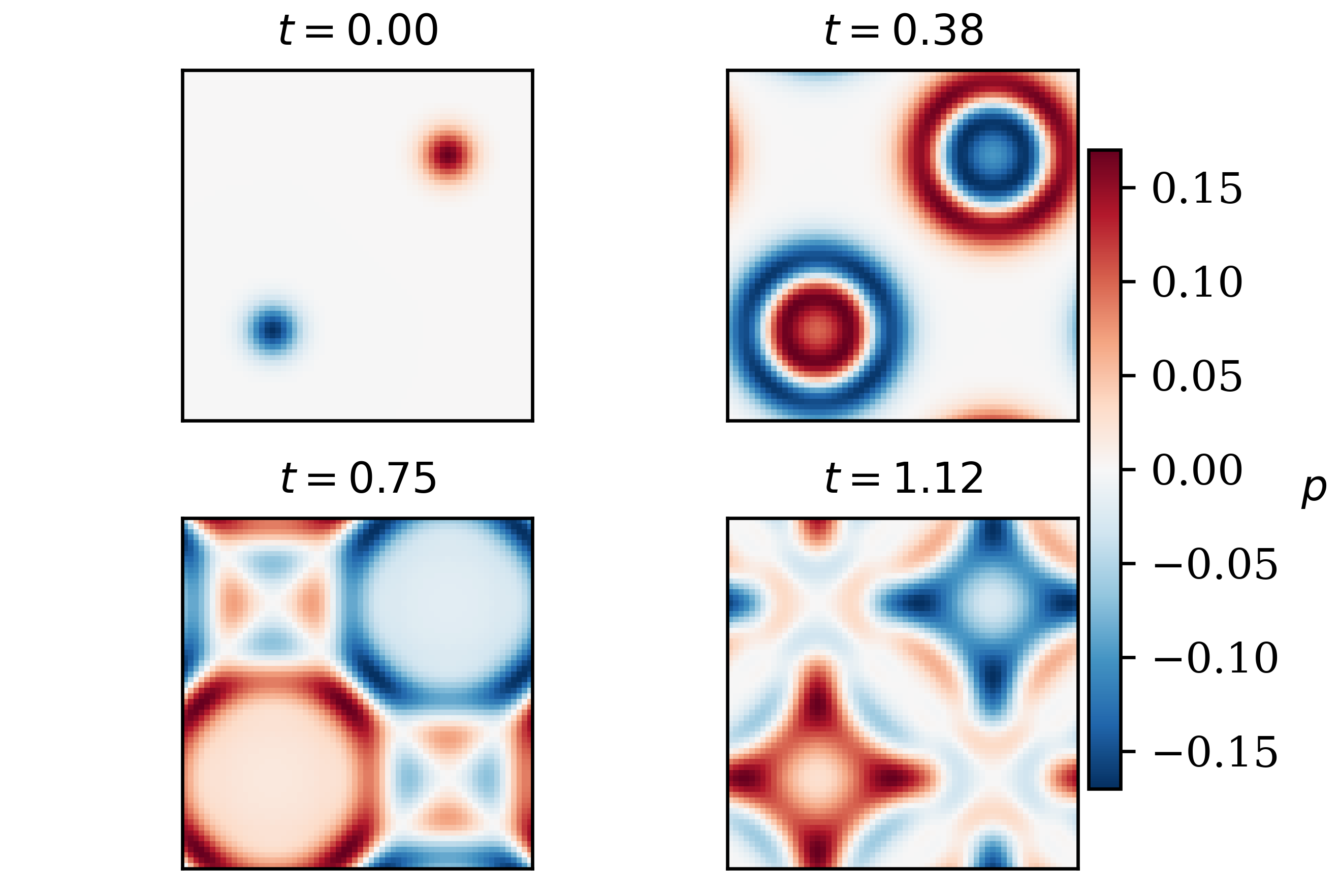}
  \caption{Snapshot of the pressure field \(p=-K(\epsilon-\gamma)\) in a \(2\times2\) periodic domain.}
  \label{fig:maxwell-snaps}
\end{figure}

%------------------------------------------------------------------
In the numerical experiments, the following parameters are chosen
\[
  K_1=K_2=1,\quad
  \rho=1,\quad
  \eta=3.4.
\]
For the initial configuration, we prepare 
the volumetric strain as  the {difference} of two Gaussian bumps
\[
  \epsilon=(x,y,0)
    =e^{-[(x-0.5)^2+(y-0.5)^2]/(2\sigma_0^2)}
    -e^{-[(x-1.5)^2+(y-1.5)^2]/(2\sigma_0^2)},
  \qquad \sigma_0=0.05 ,
\]
while \(\bm p=\gamma=0\).   Each hump emits a
cylindrical acoustic wave, and because of the periodic boundary condition,  wraps around the domain.
\cref{fig:maxwell-snaps} shows four snapshots
at various time steps, illustrating outward propagation, wrap-around
interference, and attenuation.

To benchmark the tight-binding dilation, we couple this system to an
ancilla using $F_h$
on a geometric grid~\eqref{p-geom}. For fixed $M$ and
$\delta$, we vary the post-selection point $p_\ast$ in the ancilla.
Figure~\ref{fig:lift-error} compares the pressure signal $p(t)$ at
$(x,y)=(1/4,1/4)$  for
several choices of $p_\ast$. For small $p_\ast$ the dilated trajectory
is visually indistinguishable from the reference solution over the
entire time interval, while for larger $p_\ast$ deviations appear
earlier. In the appendix we show that this behavior confirms the causal structure of the lattice embedding: where the
boundary-induced error is exponentially suppressed in
$\frac{2e K_{\max } T}{(M-j_\ast)\sinh(\delta/2)}$ with $j_\ast$ being the readout site.
The simulations are performed classically with a Runge-Kutta time integrator and a small time step so that the time
discretization error is negligible. Because the dilated Hamiltonian requires only nearest-neighbor interactions in the ancilla, we expect these dynamics to be implementable on near-term
digital quantum devices without complex connectivity overhead. The code used to generate these results is
available at Ref.~\cite{Li-qda2025}.

\begin{figure}[t]
  \centering
  \includegraphics[width=4.37in]{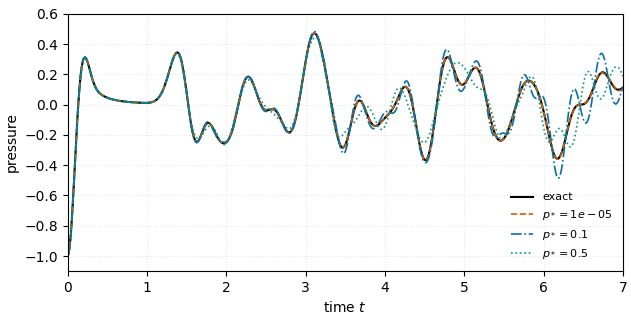}
  \caption{Pressure time series at $(x,y)=(1/4,1/4)$ for the Maxwell
viscoelastic wave. Colored curves show the dilated dynamics for varying readout points $p_\ast$. 
Smaller $p_\ast$ tracks the reference for longer times, while  larger $p_\ast$ leads to  earlier deviations, illustrating the boundary light-cone effect.
}
  \label{fig:lift-error}
\end{figure}

All these numerical simulations are performed on classical computers using fourth-order Runge-Kutta methods with a small step size to remove the time integrator error from the dilation error. Due to the sparse structure of the dilate Hamiltonian, we anticipate that these simulations can also be performed on near-term quantum devices. The code is available at \cite{Li-qda2025}.

%=====================================================================
\section{Dilation with a Bargmann--Fock Ancilla}
\label{sec:bargmann}
Jin and Liu \cite{jin2024analog} proposed an operator $F_\theta$ that uses the momentum quadrature $\hat p$ to construct the dilated Hamiltonian for analog quantum platforms.   
 Here we present a dilation scheme, where the operator $F$ and state $\kee{r}$ are naturally suited for realization on photonic quantum computing platforms.

Let $\mathcal H_{   A}=\mathscr B$ be the Bargmann–Fock space of entire
functions \cite{hall2013quantum}, with inner product
$\langle f,g\rangle=   \int_{\mathbb C}f(z)^* g(z) e^{-|z|^{2}}  \mathrm d^{2}   z$.
Here $\hat a=\partial_z$ and $\hat a^{\dagger}=z$.
For any $\theta >0 $ define
\begin{equation}
  F_\theta=\theta\bigl(z-\partial_z\bigr),\quad
  \kee{r_\theta}=\exp   \Bigl(\tfrac12 z^{2}-\tfrac z\theta\Bigr),\quad
  \brr{l}f=f(0).
\end{equation}
Because $(z-\partial_z)e^{\frac12 z^{2}-\mu z}=\mu e^{\frac12 z^{2}-\mu z}$,
choosing $\mu=1/\theta$ yields
$F_\theta \kee{r_\theta}=\kee{r_\theta}$,
so the moment identities~\eqref{moms} hold exactly and the dilation
is exact for every~$\theta$. 
As a test case we take the four-site $\mathcal PT$-SSH dimer, similar to \cite{yao2018edge},
\begin{equation}
  H=   \begin{pmatrix}
     \delta & J_{1} & 0 & 0\\
     J_{1} & 0 & J_{2} & 0\\
     0 & J_{2} & \delta & J_{1}\\
     0 & 0 & J_{1} & 0
  \end{pmatrix},  \quad 
  K=\gamma  \mathrm{diag}(1,0,1,0),
\end{equation}
with $(J_1,J_2,\delta,\gamma)=(1.0,0.6,0.3,-1/16)$ and $T=0.5$.
As a result, the dilated Hamiltonian $\widetilde H$ acts on five qumodes (one ancilla).

Toward a continuous-variable (CV) realization,
we simulate \(\widetilde H\) on a single CV mode (the ancilla) plus
four qumodes for the SSH chain using the \textsc{Strawberry Fields} package \cite{killoran2019strawberry} for photonic quantum computing platforms.   
We initialize the solution state in the vacuum, $\mathbf{x}_0=\ket{0}$, then prepare $\ket{e_r}$ with a squeezing gate, $\texttt{Sgate}!\bigl(r=\tfrac12\ln 2\bigr)$, followed by a displacement gate, $\texttt{Dgate}!\bigl(\alpha=-1/(\sqrt{2},\theta)\bigr)$. Throughout, we set $\theta = 0.5$ and truncate the Fock basis at five photons. 
We then evolve the dilated Hamiltonian with a symmetric second-order Trotter method with step size  $\Delta t=0.025$.  In each time step, we  apply an  
      on-site rotations \(R(\pm\delta    \Delta t)\), 
      beam-splitters for the off diagonals $J_1$ and $J_2$
      and the loss coupling,  implemented by
       sandwiching
      \(V(\pi/2)  {\rm CX}(-\gamma\Delta t)  V(-\pi/2)\),
      where \(V(\phi)=e^{i\phi x^{3}}\).
As an observable, we measure the ancilla with \texttt{MeasureFock} and compute
   $\rho_{\text{edge}}(t)=|x_{1}(t)|^{2}$.

As shown in Fig.~\ref{fig:cv-benchmark}, the CV data (markers) and exact solution
(solid line) shows reasonable agreement, confirming the feasibility of the
Bargmann–Fock dilation method.

\begin{figure}[hpt]
  \centering
  \includegraphics[width=0.49\linewidth]{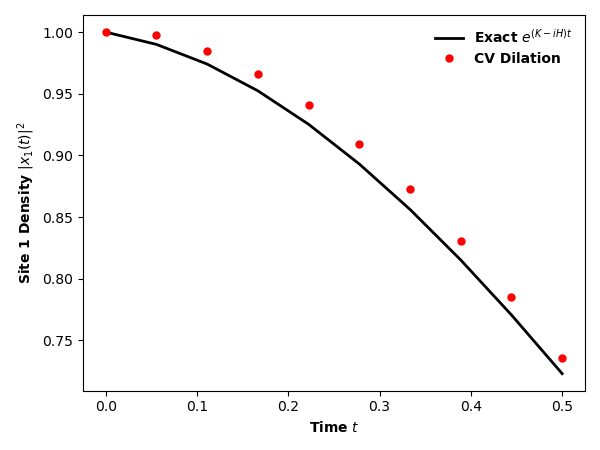}
  \caption{Edge-density dynamics for the $\mathcal PT$-SSH dimer:
           simulation using CV dilation (dots) vs.\ exact results using matrix exponential (solid).}
  \label{fig:cv-benchmark}
\end{figure}

\section{Further Discussions}

\subsection{Implementation by LCU}

Our dilation admits a linear–combination-of-unitaries (LCU) form that is
closely related to the linear combination of Hamiltonian simulation
(LCHS) technique~\cite{ACL23,ALL23}, but arises from a {different}
spectral transform (Mellin rather than Fourier).
Assume the interior generator \(F\) admits the spectral resolution
\begin{equation}\label{eq:mellin-res}
   F \;=\; - i  \int_{\mathbb R} \omega\;
        \ket{\psi_\omega} \bra{\psi_\omega}\; d\mu(\omega),
\end{equation}
so that \(iF = \int \omega  \ket{\psi_\omega} \bra{\psi_\omega}  d\mu\).
Then the dilated Hamiltonian factorises as
\[
I_A \otimes H(t) \;+\; i F \otimes K(t)
    \;=\;
   \int_{\mathbb R} \ket{\psi_\omega} \bra{\psi_\omega}
     \;\otimes\; \bigl(H(t)+\omega K(t)\bigr)\; d\mu(\omega).
\]
Consequently, for any left–right test vectors \(\brr{l}\), \(\kee{r}\) and
initial state \(\ket{\bm x_0}\),
\begin{equation}\label{eq:lchs-mellin}
\begin{aligned}
&       (\brr{l} \otimes I)\;
   \mathcal T \exp \Bigl(
     -i  \int_0^T \bigl[I_A \otimes H(s) + i F \otimes K(s)\bigr]  ds
   \Bigr)\;
   (\kee{r} \otimes\ket{\bm x_0}) \\ 
 &  \;=\;
   \int_{\mathbb R}
      \bigl[\brr{l}\psi_\omega\bigr] 
     \bigl[\bra{\psi_\omega} \kee{r}\bigr] 
     \mathcal T e^{-i \int_0^T (H(s)+\omega K(s)) ds}
     \ket{\bm x_0}  d\mu(\omega).
\end{aligned}
\end{equation}
Equation~\eqref{eq:lchs-mellin} is an LCHS–type integral representation:
after quadrature in \(\omega\), one obtains a concrete LCU over propagators
of the form \(\mathcal T e^{-i \int (H+\omega K)}\).

\medskip
For the differential operator $F$, the Mellin–wave basis
\cite{DebnathBhatta2015}
\[
   \psi_\omega(p)=p^{-1/2-i\omega},\qquad \omega\in\mathbb R,
\]
satisfies \(F \psi_\omega = -i\omega \psi_\omega\), and \(d\mu(\omega)=d\omega/(2\pi)\),
so \eqref{eq:mellin-res} is precisely the Mellin transform of \(F\).

\medskip

A fully discrete LCU form also appears after the discretization on the SBP stencil. Since the grid length
\(M=\Omega(\log(1/\epsilon))\) is modest, we can \emph{classically}
diagonalise the interior matrix \(F_h\in\mathbb C^{(M+1)\times(M+1)}\):
\begin{equation}
    F_h \;=\; - i \sum_{j=0}^{M}\omega_j\;\ket{\phi_j} \bra{\phi_j},
    \qquad
    \omega_j\in\mathbb R,\quad \{\ket{\phi_j}\}_{j=0}^{M} \text{ orthonormal}.
\end{equation}
Then
\[
I_A \otimes H(t) \;+\; i F_h \otimes K(t)
   \;=\;\sum_{j=0}^{M}
      \ket{\phi_j} \bra{\phi_j}\;\otimes\;\bigl(H(t)+\omega_j K(t)\bigr),
\]
and the time–ordered propagator decomposes as
\[
U(t)
  \;=\;\sum_{j=0}^{M}
      \ket{\phi_j} \bra{\phi_j}\;\otimes\;
      \mathcal T e^{-i\int_0^t \bigl(H(s)+\omega_j K(s)\bigr) ds}.
\]
Hence, for the dilation–consistent pair
\begin{equation}\label{eq:lhrh}
\kee{r_h}= \ket{ r_h}= \frac{1}{\sqrt{S}}\sum_{j=0}^{M}\sqrt{w_j} |j\rangle,\qquad
\brr{l_h } = \sqrt{\frac{S}{S_W}} \bra{l_h}, \;  \bra{l_h}= \frac{1}{\sqrt{S_W}}\sum_{j=0}^{j_\ast}\sqrt{w_j} \langle j|,
\end{equation}
with normalizing constants $S_W=\sum_{j=0}^{j_\ast}w_j,\qquad S=\sum_{j=0}^{M}w_j$
we obtain the explicit LCU expansion
\begin{equation}
\ket{\bm x(t)}
\;=\;
\sum_{j=0}^{M}
      \underbrace{\langle l_h|\phi_j\rangle \langle \phi_j|r_h\rangle}_{=:c_j}\;
      \mathcal T e^{-i\int_0^t \bigl(H(s)+\omega_j K(s)\bigr)ds}\;
      \ket{\bm x_0}
\;=\;\sum_{j} c_j  U_j(t)\ket{\bm x_0}.
\end{equation}
By Cauchy–Schwarz with the orthonormal basis \(\{\ket{\phi_j}\}\),
\[
\| \bm c \|_1
\;\le\; \|\ket{l_h}\| \|\ket{r_h}\|
\;=\; \frac{\sqrt{S}}{\sqrt{S_W}}
\;=\; \frac{1}{\sqrt{P_{\rm win}(j_\ast)}},
\qquad
P_{\rm win}(j_\ast):=\frac{S_W}{S}.
\]
Thus a direct (single–shot) LCU implementation has success amplitude
scaling like \(\sqrt{P_{\rm win}}\); if the window weight \(P_{\rm win}\) is tiny,
the LCU \(\ell_1\)–weight \(\|\bm c\|_1\) is large.
A standard remedy is time segmentation: split \([0,T]\) into
\(L\) segments of length \(\tau=\Theta(1/K_{\max})\), apply LCU on each
segment (so that \(\ell_1\)–weights per segment are \(\co(1)\)),
use robust oblivious amplitude amplification (OAA) per segment, and
coherently \emph{restore} the ancilla back to \(\ket{r_h}\) at the end of each
segment. In this way the per–segment LCU \(\ell_1\)–weight
\(\mathsf{L}_{\rm seg}\) can be kept \(\le 2\) (cf.\ \cite[Sec.~3]{low2018hamiltonian}),
yielding an overall near–optimal scaling when combined with the
Dyson–series interaction–picture framework in the previous subsection.

\subsection{ODEs with transient or persistent growth}
Most earlier quantum algorithms target \emph{dissipative} ODEs where
$K\prec 0$, so solutions are contractive (e.g., \cite{jin2023quantum,ALL23}).
Even for an autonomous \emph{stable} system
$\bm x'(t)=A\bm x(t)$ whose spectrum lies in the closed left half–plane
(and purely imaginary eigenvalues are simple), non-normality can cause
substantial \emph{transient} growth when $K$ is not negative definite.

\begin{figure}[htp]
  \centering
  \includegraphics[width=0.75\linewidth]{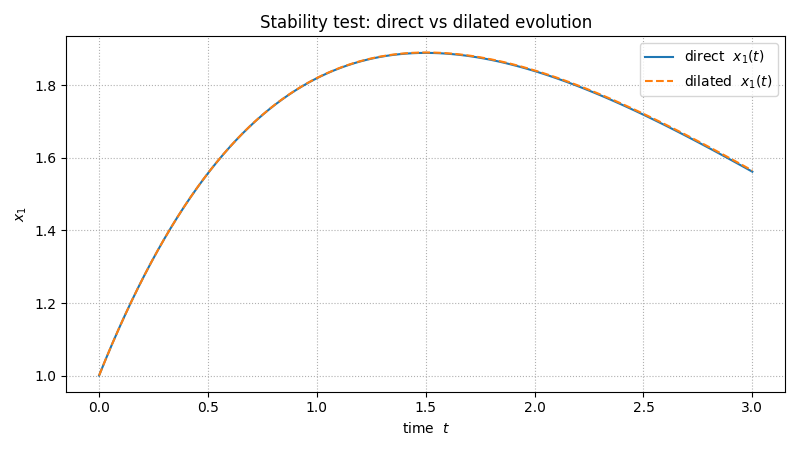}
  \caption{Exact and dilated solutions of \eqref{eq:stable-dilation}.}
  \label{fig:stable-ODE}
\end{figure}

As a representative example, consider
\begin{equation}\label{eq:stable-dilation}
  A \;=\;
  \begin{bmatrix}
    -\tfrac12 & 1 \\[2pt]
    0 & -\tfrac12
  \end{bmatrix},
  \qquad
  \bm x_0=\begin{bmatrix}1\\[2pt]1\end{bmatrix}.
\end{equation}
The exact solution exhibits transient growth with
$\|\bm x(t)\|\sim t e^{-t/2}$ before decaying.
We discretize the dilation using the second–order diagonal–norm SBP split form
on a \emph{geometric mesh} \(p_j=\exp[-(1-j/M)]\) with \(M=10\) points
(\(\delta=1\), i.e.\ \(\Lambda=M\)).
We set \(\theta=2\) (\(\beta=0\)), prepare the ancilla in
\(\ket{r_h}\propto W^{1/2}\mathbf 1\), and post–select at
\(j_\ast=8\) so that \(\langle l_h|r_h\rangle=1\).
Both the physical ODE and the dilated system are integrated in time by a
fourth–order Runge–Kutta scheme up to \(t_{\mathrm{final}}=3\) using
\(n_{\mathrm{steps}}=400\) (step size \(7.5\times 10^{-3}\)),
making the time–integration error negligible compared to the
spatial/discretization effects.
The ancilla–projected trajectory \(x_1(t)\) from the dilated evolution matches the exact solution. This provides direct numerical evidence (complementing our main theorems) that the dilation framework also applies to ODEs with transient instability.

\medskip

The dilation also applies to \emph{unstable} ODEs with persistent growth. As a simple test, take
\begin{equation}\label{eq:unstable-ode}
  A=\begin{bmatrix}\tfrac12 & 1 \\[2pt] 0 & \tfrac12\end{bmatrix},
  \qquad
  \bm x_0=\begin{bmatrix}1\\[2pt]1\end{bmatrix}.
\end{equation}
Figure~\ref{fig:unstable-ODE} shows excellent agreement between the exact and
dilated solutions, with the latter obtained with a similar dilation scheme as above. 

\begin{figure}[htp]
  \centering
  \includegraphics[width=0.75\linewidth]{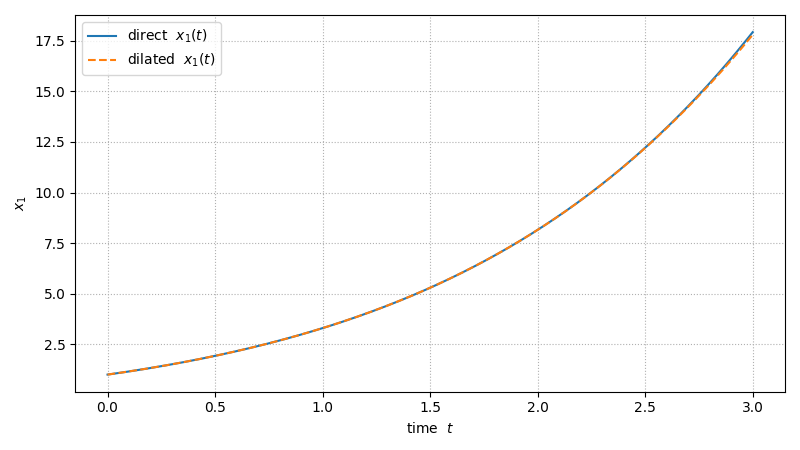}
  \caption{Exact and dilated solutions of \eqref{eq:unstable-ode}.}
  \label{fig:unstable-ODE}
\end{figure}

\section{Summary and outlook}
This paper presented a general dilation scheme that encodes the dynamics associated with an ODE or PDE into a unitary dynamics. We propose to wrap such dilation schemes through three elements, a skew Hermitian operator $F$, an encoding element $\kee{r}$, and an evaluation functional $\brr{l}$.  The seven concrete families detailed in the main text, together with flexible parameter $\theta$ and unitary transformation among the triple, illustrate the breadth of this new framework, which allows the dilation to be {co-designed} with the target hardware (gate-based, continuous-variable, analogue) and with the structure of a given application.

Nonlinear problems are generally difficult for quantum computing algorithms. One approach that has been extensively studied is  the Carleman embedding method  \cite{liu2021efficient,liu2023efficient,krovi2023improved,wu2025quantum,brustle2024quantum}, which can be viewed as a dilation as well, where $\kee{r}$ corresponds to monomials, forming a nonlinear transformation, and $\brr{l}$ can be interpreted as the evaluation of the first monomial. It is possible that the current framework can produce other dilation schemes for nonlinear problems, and it is an interesting future direction.  

\bigskip 
\paragraph*{Acknowledgement.}   This research is supported by the NSF Grant DMS-2411120.

\bibliographystyle{abbrv}
%\bibliography{references}

\appendix

\section{Lightcone property of the finite difference operator}
\label{proof-thm1}

This section is dedicated to the error bound \cref{thm:second-order-error}, which we will prove by separating the approximation error of the differential operator and the boundary error.

The dilation method involves the following partial differential equation,
\begin{equation}\label{eq: pde}
     u_t =- \theta \kappa(t) (p u_p + \frac12 u), \quad u(t,1)= u_R(t),  
\end{equation}
with solution $u(t,p): [0,+\infty) \otimes [0,1] \mapsto \mathbb{R} $. $\theta >0$ is a numerical parameter that we will choose to optimize our quantum algorithms later. 
This is a first-order PDE,  and the solution can be constructed by the method of characteristics. When $\kappa(t)$ is positive,  the solution propagates to the left and $u_R(t)$ acts as the boundary condition.  For convenience, we write 
\begin{equation}\label{eq: pde'}
  u_t =-  \kappa(t) F_\theta u, \quad F_\theta = \theta F, 
\end{equation}
with the partial differential operator given by, for any differentiable function $v(p)$,
\begin{equation}\label{eq: Fvp}
     Fv(p):= \left( p  \partial_p v(p)+\tfrac12 v(p)\right). 
\end{equation}

The time-dependent coefficients $-\kappa(t) $ will be generalized to a Hermitian negative definite matrix $K(t)$ later in this section. 

\medskip 

To obtain a direct numerical method, we partition the unit interval into $M$ subintervals, and find approximation solutions on a uniform grid with grid size denoted by $h$,
\begin{equation}\label{eq: grid}
    p_i=i  h, h=\dfrac1M,   \quad 0\leq i \leq M.
\end{equation}

To construct an approximation to the solution of \cref{eq: pde'}, we construct a difference operator on the grid, as follows, denoting $v_i:=v(p_i)$,
\begin{equation}\label{eq: Fh2}
\begin{aligned}
(F_h\bm v)_0 &=  \frac{1}{4\sqrt2} v_{1},\\[4pt]
(F_h\bm v)_1 &=  \frac34 v_{2}-\frac{1}{4\sqrt2} v_{0},\\[4pt]
(F_h\bm v)_i &=  \frac{p_{i+1}+p_i}{4h} v_{i+1}
            -\frac{p_i+p_{i-1}}{4h} v_{i-1},
            \qquad 2\le i\le M-1,\\[4pt]
(F_h\bm v)_M &= -\frac{p_M+p_{M-1}}{4h} v_{M-1} .
\end{aligned}
\end{equation}
The coefficients $\frac{1}{4\sqrt{2}}$ come from the quadrature approximation in the summation by parts methods \cite{Strand1994,MattssonNordstrom2004}.

The following property can be easily verified by using integration/summation by parts. 
\begin{lemma}\label{F-herm}
    With zero boundary condition, $v(1)=0$, both operators $F$ in \cref{eq: Fvp} and \cref{eq: Fh2} are skew Hermitian. 
\end{lemma}

In the next few sections, we explain several approximation properties of $F_h $ in approximation the solution of the PDE \eqref{eq: pde'}. 

\subsection{Finite–Propagation Property of the Semi-Discrete Scheme}
\label{sec:finite-propagation}

One important property of the PDE \cref{eq: pde} is the finite propagation speed and finite domain of influence. According to the method of characteristics, $-\theta \kappa(t)$ is regarded as the propagation speed of the solution moving to the left. For example, if $\kappa(t)$ is always bounded by $\tfrac{1}{4T}$, then the part of the initial condition $u(0,p)$ that lies in $p\in (\tfrac34,1)$ would have no impact on the solution $u(T,p)$ for $p \in (0, \tfrac12).$ 

As a simple demonstration, consider the first–order linear PDE
\[
u_t  = - p  u_p - \tfrac12  u,
\qquad 0<p<1,\quad t\ge0. 
\]
Its behaviour is completely transparent through the \emph{method of
characteristics}.  One can find the solution by following a family of trajectories \((p(t),t)\), satisfying $\frac{dp}{dt}=-  p,$ and along these trajectories, the solution of the above PDE is reduced to the solution of the ODEs (by chain rule)
\[
\frac{du}{dt}=-  \tfrac12  u. 
\]
The first equation integrates to an exponential
\[
p(t)=p_0  e^{-(t-t_0)}  \Longrightarrow  
t-t_0=\ln\frac{p_0}{p(t)} . 
\]
Thus every “signal’’ injected at some point \(p_0\) travels
\emph{inward} with speed proportional to the current \(p\); it needs the
finite time
\[
T(p_0   \to p_\ast)=\ln   \frac{p_0}{p_\ast}
\]
to reach a location \(p_\ast<p_0\).
Meanwhile, the second ODE for $u$ gives exponential damping along the same curve,
\[
u(p(t),t)=u(p_0,t_0)  
          \exp   \Bigl[-\tfrac12\bigl(t-t_0\bigr)\Bigr]. 
\]

Take \(p_0=1\) (the boundary) and impose \(u(1,t)=1\).
These equations describe the disturbance at the boundary that first hits the interior point
\(p_\ast\) exactly at time 
\begin{equation}\label{arriveT}
    T_\text{arrive}=\ln\frac1{p_\ast},
\end{equation}
no sooner: One remarkable consequence is that such disturbance takes infinite time to arrive at $p=0.$

We now establish a finite-propagation-speed estimate for the discrete scheme. The bound shows that the influence of the right-hand boundary condition at $
p=1$  decays exponentially as one moves into the interior of the grid, playing the same role for our numerical solution that the Lieb–Robinson bound \cite{Lieb1972} plays in quantum lattice systems.

\begin{lemma}[Finite propagation speed bound]
\label{lem:fps-sharp}
Fix $\theta>0$ and assume the modulation
$0\le\kappa(t)\le K_{\max}$ on the whole interval $[0,T]$.
Let $p_\ast\in(0,1)$ be an interior abscissa and suppose the
travel-time condition
\begin{equation}\label{cfl-sharp}
   e  \theta  K_{\max}  T
    < 
   1- {p_\ast}
\end{equation}
holds.
Start from the canonical right-boundary vector
\(
\bm u^{R}=(0,\dots,0,1)^{T} = \ket{M}.
\)
Then for every $t\in[0,T]$ and every grid index $i$ with $p_i\le p_\ast$
the propagated amplitude satisfies
\[
\bigl|  \bra{i}  e^{-   \int_{0}^{t}\theta\kappa(s)  ds  F_h}  \ket{\bm u^{R}}\bigr|
    \le 
         \frac{ \left( \frac{e  \theta  K_{\max}  T}{ 1- {p_\ast}} \right)^m }{1-\frac{e  \theta  K_{\max}  T}{ 1- {p_\ast}}},
   \qquad m=M-i .
\]
\end{lemma}

\begin{proof}
Write $F_h=B_{+}-B_{-}$ where
\((B_{+})_{k,k+1}=(p_{k+1}+p_k)/(4h)\) and
\((B_{-})_{k,k-1}=(p_{k}+p_{k-1})/(4h)\).
Each product of $k$ factors moves information at most $k$ nodes, hence
$\langle i|F_h^{k}|M\rangle=0$ whenever $k<m$.

Because $(p_r+p_s)\le2$ and every row has two non-zero entries,
\[
\|F_h\|_{\infty}\le\frac{1}{h},
\quad\Longrightarrow\quad
\|F_h^{k}\|_{\infty}\le h^{-k}.                               
\]
Combining these observations, we have 
\[
|\langle i|e^{-   \int_0^t\theta\kappa(s)ds  F_h}|M\rangle|
  \le 
 \sum_{k=m}^{\infty}\frac{(\theta K_{\max}t/h)^{k}}{k!}.
\]

Using Stirling’s bound \(k!\ge(k/e)^{k}\) we simplify each term to
\(\dfrac{(\theta K_{\max}t/h)^{k}}{k!}
      \le\bigl(\dfrac{e  \theta K_{\max}t}{h  k}\bigr)^{k}\).
For $k=m$ and $h  m\ge 1-  p_\ast$ (because $p_i\le p_\ast$) this ratio is
bounded by
\[
\frac{e  \theta K_{\max}t}{h  m}
  \le 
 \frac{e  \theta K_{\max}t}{1-p_\ast}
 <1,
\]
where the last inequality is the travel-time hypothesis
\eqref{cfl-sharp}.
Summing the geometric tail yields the stated exponential bound.
\end{proof}

\bigskip 

This finite propagation estimate can be directly generalized to bound the error due to the homogeneous boundary condition that helps to ensure the Hermitian property of the dilated Hamiltonian. To the error from the interior, let us first consider the following dynamics: 
\begin{equation}\label{eq: chi-error1}
    \frac{d}{dt}\chi(t)=\bigl(-i I_A \otimes H(t)+\theta F_h \otimes K(t)\bigr)\chi(t)
+ \alpha  \ket{M} \otimes K(t) \bm  x(t),\quad \chi(0)=0.
\end{equation}
The first term on the right hand side in \cref{eq: chi-error1} is the dilated Hamiltonian, while the second term, which only concentrates at the boundary, comes from the effect of dropping the boundary condition with some positive constant $\alpha$.  In particular, it carries the solution of the ODE: $\bm x(t).$ How \cref{eq: chi-error1} emerges from the global error analysis will be elaborated in the next section.

\begin{lemma}[Finite propagation speed bound for the dilated system]
\label{lem:fps-chi}
Fix $\theta>0$ and
$K_{\max}:=\max_{t\in[0,T]}\|K(t)\|$.
Let $p_\ast\in(0,1)$ be an interior abscissa and assume the travel-time condition
\eqref{cfl-sharp} from Lemma~\ref{lem:fps-sharp},
\[
e \theta K_{\max} T  <  1-p_\ast.
\]
Fix an interior evaluation point $p_{ *}=p_i\in(0,1)$ and set $m:=M-i= (1-p_\ast)/h$.
Consider $\chi$ governed by
\begin{equation}\label{eq:chi-error}
    \frac{d}{dt}\chi(t)=\bigl(-i I_A \otimes H(t)+\theta F_h \otimes K(t)\bigr)\chi(t)
     + \alpha \ket{M} \otimes K(t) \bm x(t),\qquad \chi(0)=0.
\end{equation}
Then, for every $T\in[0,\infty)$,
\begin{equation}\label{eq:chi-final-bound}
\bigl\|(\bra{i} \otimes I) \chi(T)\bigr\|
 \le 
\alpha K_{\max} X(T) \|\bm x(0)\| T\cdot
\frac{\bigl(\frac{e \theta K_{\max} T}{ 1-p_\ast }\bigr)^{m}}{1-\frac{e \theta K_{\max} T}{ 1-p_\ast }},
\end{equation}
where $X(T):=\max_{t\in[0,T]}\|\bm x(t)\|/\|\bm x(0)\|$.
Consequently, if $ \varrho:=\frac{e \theta K_{\max} T}{1-p_\ast}<1$, then for any $\epsilon>0$
it suffices to choose
\[
m \ge \left\lceil
\frac{\ln \bigl(\alpha K_{\max} X(T) \|\bm x(0)\| T/\bigl(\epsilon(1-\varrho)\bigr)\bigr)}{\ln(1/\varrho)}
\right\rceil
\]
to guarantee $\bigl\|(\bra{i} \otimes I) \chi(T)\bigr\|\le \epsilon$.
\end{lemma}

\begin{proof}
We first eliminate the $I_A \otimes H$ term in \cref{eq:chi-error} by working in the interaction picture. 
Let $U_H(t)$ solve $U_H'(t)=-i H(t) U_H(t)$ with $U_H(0)=I$,
and define the interaction-picture quantities
\[
K_I(t):=U_H(t)^\dagger K(t)U_H(t),\qquad
\bm x_I(t):=U_H(t)^\dagger \bm x(t),\qquad
\chi_I(t):=(I_A \otimes U_H(t)^\dagger) \chi(t).
\]
Since $U_H(t)$ is unitary, $\|K_I(t)\|=\|K(t)\|$ and $\|\bm x_I(t)\|=\|\bm x(t)\|$.

Then \eqref{eq:chi-error} becomes the inhomogeneous linear ODE
\begin{equation}\label{eq:chiI-ODE}
\frac{d}{dt}\chi_I(t)
=\bigl(\theta F_h \otimes K_I(t)\bigr)\chi_I(t)
 + \alpha \ket{M} \otimes K_I(t) \bm x_I(t),
\qquad \chi_I(0)=0.
\end{equation}

Let $W(T,t)$ denote the fundamental (propagator) solution of the homogeneous part,
\[
\partial_T W(T,t)=\bigl(\theta F_h \otimes K_I(T)\bigr)W(T,t),\qquad W(t,t)=I.
\]
Variation of constants gives
\begin{equation}\label{eq:chiI-variation}
\chi_I(T)
=\alpha\int_0^T W(T,t) \ket{M} \otimes K_I(t) \bm x_I(t) dt.
\end{equation}

\medskip

Because $F_h$ is time-independent, the Dyson (time-ordered) series for $W(T,t)$ reads
\begin{equation}\label{eq:dyson-W}
W(T,t)
=I+\sum_{k=1}^{\infty}\theta^{k} 
\underbrace{\int_{t\le t_1\le\cdots\le t_k\le T}
\bigl(F_h \otimes K_I(t_k)\bigr)\cdots\bigl(F_h \otimes K_I(t_1)\bigr) dt_1\cdots dt_k}_{=:~\mathcal{J}_k(T,t)}.
\end{equation}
Since $F_h$ commutes with itself,
\[
\mathcal{J}_k(T,t)=F_h^{ k} \otimes \Bigl(\int_{t\le t_1\le\cdots\le t_k\le T}
K_I(t_k)\cdots K_I(t_1) dt_1\cdots dt_k\Bigr)
=F_h^{ k} \otimes J_k(T,t),
\]
where, with the convention $J_0(T,t):=I$,
\begin{equation}\label{eq:Jk-def}
J_k(T,t):=\int_{t\le t_1\le\cdots\le t_k\le T}
K_I(t_k)\cdots K_I(t_1) dt_1\cdots dt_k.
\end{equation}
Thus,
\begin{equation}\label{eq:W-factorized}
W(T,t)=\sum_{k=0}^{\infty}\theta^{k} F_h^{ k} \otimes J_k(T,t).
\end{equation}

\medskip

Insert \eqref{eq:W-factorized} into \eqref{eq:chiI-variation}:
\[
\chi_I(T)
=\alpha\sum_{k=0}^{\infty}\theta^{k} F_h^{ k}\ket{M} \otimes 
\Bigl(\int_0^T J_k(T,t) K_I(t) \bm x_I(t) dt\Bigr).
\]
Apply $(\bra{i} \otimes I)$ and use the band structure of $F_h$:
\begin{equation}\label{eq:band-cutoff}
\langle i|F_h^{ k}|M\rangle=0\quad\text{for all }k<m,\qquad m:=M-i,
\end{equation}
so only the terms with $k\ge m$ remain:
\begin{equation}\label{eq:chiI-tail}
(\bra{i} \otimes I) \chi_I(T)
=\alpha\sum_{k=m}^{\infty}\theta^{k} \langle i|F_h^{ k}|M\rangle 
\Bigl(\int_0^T J_k(T,t) K_I(t) \bm x_I(t) dt\Bigr).
\end{equation}

\medskip

We recall the uniform bound from Lemma~\ref{lem:fps-sharp}:
since each row of $F_h$ has at most two nonzero entries of size $\le 1/(2h)$,
\[
\|F_h\|\le \frac{1}{h}\quad\Longrightarrow\quad
\|F_h^{ k}\|\le h^{-k}.
\]
Also, $\|K_I(t)\|\le K_{\max}$ and $\|\bm x_I(t)\|\le X(T)$.
For the time-ordered integrals \eqref{eq:Jk-def},
\[
\|J_k(T,t)\|
 \le \int_{t\le t_1\le\cdots\le t_k\le T}       \|K_I(t_k)\cdots K_I(t_1)\| dt_1\cdots dt_k
 \le \frac{(K_{\max})^{k}}{k!} (T-t)^{k}.
\]
Therefore,
\[
\biggl\|\int_0^T J_k(T,t) K_I(t) \bm x_I(t) dt\biggr\|
\le\int_0^T \frac{(K_{\max})^{k}}{k!}(T-t)^k K_{\max} X(T)  dt
=\frac{(K_{\max}T)^{k+1}}{(k+1)!} K_{\max} X(T).
\]
Combining with $\|\langle i|F_h^{ k}|M\rangle\|\le \|F_h^{ k}\|_{\infty}\le h^{-k}$,
\eqref{eq:chiI-tail} yields
\[
\bigl\|(\bra{i} \otimes I) \chi_I(T)\bigr\|
\le
\alpha K_{\max} X(T) T 
\sum_{k=m}^{\infty}\frac{\bigl(\theta K_{\max} T/h\bigr)^{k}}{(k+1)!}.
\]
Use $(k+1)!\ge k!$ to get the simpler tail
\[
\sum_{k=m}^{\infty}\frac{\bigl(\theta K_{\max} T/h\bigr)^{k}}{(k+1)!}
 \le \sum_{k=m}^{\infty}\frac{\bigl(\theta K_{\max} T/h\bigr)^{k}}{k!}.
\]
Now apply Stirling’s bound $k!\ge (k/e)^k$ as in Lemma~\ref{lem:fps-sharp}:
\[
\frac{\bigl(\theta K_{\max} T/h\bigr)^{k}}{k!}
 \le \Bigl(\frac{e \theta K_{\max} T}{h k}\Bigr)^{k}.
\]
For $k\ge m$ and $hm\ge 1-p_\ast$ (since $p_i\le p_\ast$), we obtain the uniform ratio
\[
\Bigl(\frac{e \theta K_{\max} T}{h k}\Bigr)^{k}
 \le 
\Bigl(\frac{e \theta K_{\max} T}{ 1-p_\ast }\Bigr)^{k}
=: \varrho^{ k},
\qquad \varrho=\frac{e \theta K_{\max} T}{1-p_\ast}<1
\]
by the travel-time hypothesis \eqref{cfl-sharp}. Summing the geometric tail gives
\[
\sum_{k=m}^{\infty}\frac{\bigl(\theta K_{\max} T/h\bigr)^{k}}{k!}
 \le \frac{\varrho^{ m}}{1-\varrho}.
\]
Inserting this into the previous bound yields \eqref{eq:chi-final-bound} for
$\|(\bra{i} \otimes I) \chi_I(T)\|$, and since the physical and interaction pictures differ by a unitary on the second tensor factor, the same bound holds for $\|(\bra{i} \otimes I) \chi(T)\|$.
\end{proof}

\subsection{Lightcone property on the geometric grid}\label{sec: CFL-geom}

Let $p_j=e^{-(M-j)\delta}$ with fixed $\delta>0$, and let $F_h$ be the split--form
Euclidean--skew tridiagonal with off--diagonals
\[
f_j=\frac{1}{4\sinh(\delta/2)}\times
\begin{cases}
\sqrt{1+e^{-\delta}}, & j=0,\\
1, & 1\le j\le M-2,\\
\sqrt{1+e^{\delta}}, & j=M-1.
\end{cases}
\]
We let the uniform interior weight be $f_{\rm int} =  \frac{1}{4\sinh(\delta/2)}.$

\begin{lemma}
\label{lem:geom-CFL-bulk}
Let $F_h$ be the split–form tridiagonal matrix from the SBP discretization on the geometrically scaled grid with $\delta \geq 1$. $\forall j_\ast \in [M] $, let $m=M-j_\ast.$  Then
\begin{equation}\label{eq:geom-CFL-bulk}
\bigl|\langle j_\ast|\chi(T)|M\rangle\bigr|
\;\le\;
\frac{C(\delta) \varrho^{\,m}}{1-\varrho^{\,2}},
\qquad
\varrho_:=\frac{ e \theta K_{\max} T }{4m \sinh \frac{\delta}{2}}
\end{equation}
where $C(\delta)= \max\{f_0, f_{M-1}\}$.  $C(\delta) \leq \sup_{\delta \geq 1} \{ f_{M-1} \} = \frac{\sqrt{1+e}}{4\sinh(1/2)} \approx 0.926$. 
\end{lemma}

\begin{proof}[Proof sketch]

Write the split–form SBP chain as
\[
F_h \;=\; \sum_{j=0}^{M-1} f_j\bigl(\,\ket{j} \bra{j{+}1}\;-\;\ket{j{+}1} \bra{j}\,\bigr)
\;=\;A-A^\dagger,\qquad
A:=\sum_{j=0}^{M-1} f_j\,\ket{j} \bra{j{+}1}.
\]
Thus $A$ lowers the site index by one (a left hop) with weight $f_j$, while $-A^\dagger$
raises the index by one (a right hop) with weight $f_j$ and an extra minus sign.
Fix $j_\ast\le M$ and set $m:=M-j_\ast$.
For every integer $k\ge m$,   $\langle j_\ast|F_h^{\,k}|M\rangle$ constitutes a  sum  over all length-$k$ nearest–neighbour walks on the line
$\{0,1,\dots,M\}$ that start at $M$ and end at $j_\ast$; each appearance of an oriented
edge $e=(\ell\leftrightarrow \ell{+}1)$  contributes a factor $f_j$. 

Expand $\chi(T) $ in Dyson series and consider again $F_h^k$. Only $k\ge m$ with $k\equiv m\ (\mathrm{mod}\ 2)$
contribute to $\langle j_\ast|F_h^k|M\rangle$. For any admissible walk of length $k$ from $M$ to $j_\ast$,
its edge–weight monomial is a product of $k$ factors drawn from $\{f_{\rm int},f_{M-1}, f_{0} \}$. Notice that $f_{M-1}, f_{0} \leq C(\delta)$, and
 the path must cross the boundary edge $M$ at least once.
The number of such walks is at most $\binom{k}{(k+m)/2}\le 2^k$, so
\[
\bigl|\langle j_\ast|F_h^k|M\rangle\bigr|
\ \le\ C(\delta) f_{\rm int}^{k}.
\]
Hence
\[
\left|\frac{(\theta K_{\max} T)^k}{k!}\,\langle j_\ast|F_h^k|M\rangle\right|
\ \le\
C(\delta) \frac{(\theta K_{\max} T f_{\rm int})^{k}}{k!}
\ \le\
C(\delta) \Bigl(\frac{ e \theta K_{\max} T f_{\rm int}}{k}\Bigr)^{k},
\]
using Stirling $k!\ge (k/e)^k$. The rest is the same as the previous Lemmas. 
\end{proof}

\section{Approximation property of the SBP finite difference}

In this section, we consider choosing the vector $\kee{r} \propto p^\beta, $ with $\beta >0$. In this case, $ \theta F_h \ket{r_h} - \ket{r_h} $ might not be zero at each grid point in the interior. Therefore, we have to take into account such local error in the analysis.

\subsection{Local consistency error }
\label{sec:SBP-local-split}

One step in establishing the global accuracy is estimating the local consistency error.

\begin{lemma}
\label{lem:SBP-local-split}
Let \(v\in C^{3}[0,1]\).  Define the \emph{global} constants
\[
M_3:=\norm{v''(p) + \frac{2p}{3} v'''(p) }_{L^{\infty}(0,1)},
\]
Then the following bounds hold, 
  \begin{equation}
       \bigl|  (F_h\bm v)_i-(Fv)(p_i)\bigr|
       \le \frac{1 }{4}  h^{2}  M_3,   \forall 1\le i\le M-1, 
       \label{local-error}
  \end{equation}
Furthermore, if $v(0)=0,$ then $(F_h\bm v)_1-(Fv)(p_1) =\co(h^2) $; and if 
$v(0)=v'(0)=0,$ then $(F_h\bm v)_0-(Fv)(p_0) =\co(h^2) $.

\end{lemma}

\begin{proof}

Using $p_{i\pm1}=p_i\pm h$ we rewrite the stencil coefficients:
\[
  \frac{p_{i+1}+p_i}{4h}= \frac{p_i}{2h}+\frac14,
  \quad
  \frac{p_i+p_{i-1}}{4h}= \frac{p_i}{2h}-\frac14.
\]
Introduce $A:=\frac{p_i}{2h}+\frac14$ and $B:=\frac{p_i}{2h}-\frac14$,
so that $A-B=\tfrac12$ and $A+B=\tfrac{p_i}{h}$.

A Taylor expansion of \eqref{eq: Fh2} about $p_i$ gives
\[
\begin{aligned}
  v(p_{i+1})&= v(p_i) + h v'(p_i) + \tfrac{h^{2}}{2}v''(p_i)
            + \tfrac{h^{3}}{6}v'''(\xi_i),\\
  v(p_{i-1})&= v(p_i) - h v'(p_i) + \tfrac{h^{2}}{2}v''(p_i)
            - \tfrac{h^{3}}{6}v'''(\eta_i),
\end{aligned}
\]
with $\xi_{i}, \eta_i \in[p_{i-1},p_{i+1}]$.

Thus, a direct substitution, together with a mean value theorem for the third-order derivatives, yields,
\[
\begin{aligned}
 (F_h v)_i
  &=(A-B)v(p_i) + h(A+B)v'(p_i) + \tfrac{h^{2}}{2}(A-B)v''(p_i)
    + \tfrac{h^{3}}{6}(A+B)v'''(\xi)  \\
  &=\frac12 v(p_i) + p_i v'(p_i)
    +\frac{h^{2}}{4}v''(p_i) + \frac{p_i h^{2}}{6}v'''(\xi) 
\end{aligned}
\]

Now we move to the node $i=1,$
\[
  (F_h\bm v)_1
     = \frac34  v_{2} - \frac{1}{4\sqrt2}  v_{0},
     \qquad p_1=h .
\]

Taylor expand \(v_2\) and \(v_0\) about \(p=h\):
\[
\begin{aligned}
v_{2} &= v(h)+h  v'(h)+\tfrac12h^{2}v''(h)+\tfrac16h^{3}v'''(\xi_2),\\
v_{0} &= v(h)-h  v'(h)+\tfrac12h^{2}v''(h)-\tfrac16h^{3}v'''(\xi_0),
\end{aligned}
\qquad
\xi_{0,2}\in(0,2h).
\]
Insert these into the formula and collect like terms:
\[
\begin{aligned}
(F_h\bm v)_1
  &=c_0 v(h)
     +c_1 h  v'(h)
     +\tfrac{c_0}2 h^{2}v''(h)
     +\tfrac{c_1}{6}h^{3} v'''(\xi).
\end{aligned}
\]
Here we introduced  
\(c_0=\frac34-\frac1{4\sqrt2}, 
 c_1=\frac34+\frac1{4\sqrt2}\) and applied the mean value theorem to $v'''$.

 Now the error becomes, 
\[
\begin{aligned}
 \left( (F_h\bm v)_1-(Fv)(h) \right)
  &= (c_0-\tfrac12)v(h)+(c_1-1)h  v'(h)
     +\tfrac{c_0}2  h^{2}v''(h)+\tfrac{c_1}{6}h^{3} v'''(\xi) \\[4pt]
  &= \frac{1-1/\sqrt2}{4}  [v(h)-h  v'(h)]
     +\tfrac{c_0}2  h^{2}v''(h)+\tfrac{c_1}{6}h^{3} v'''(\xi) \\[4pt] 
  &=    (c_0 - \tfrac14)   h^{2}v''(\eta)+\tfrac{c_1}{6}h^{3} v'''(\xi).   
\end{aligned}
\]
Here we used a second-order Taylor expansion gives
\(v(h)-h  v'(h)=\tfrac12h^{2}v''(\eta)\) (\(\eta\in(0,h)\)).

For the boundary node $i=0.$
\((F\bm v)_0=\tfrac{1}{2\sqrt2} v(h)
   =\tfrac{h^{2}}{4\sqrt2} \theta v''(\zeta)\),
again \(\co(h^{2})\).

Combining the three cases yields the stated bound.

\end{proof}

 \begin{corollary}\label{local-error-g}
Assume $0< \theta \leq \tfrac{2}{7}$. Let 
\begin{equation}\label{theta2beta}
    \displaystyle \beta=\frac1\theta-\frac12,
\end{equation}
\(g(p)=p^{\beta}\).  Then, for \(1\le i\le M-1\),
\begin{equation}
   \theta \bigl| (F_h\bm g)_i-  (Fg)(p_i)\bigr|
    \le 
   C(\theta )  h^2 p_{i+1}^{\beta-3},
\quad
C(\theta)=\frac{\theta }{6}\beta(\beta-1)(\beta-2).
\end{equation}
 \end{corollary}
\begin{proof}
For \(g(p)=p^{\beta}\) one has
\(g'''(p)=\beta(\beta-1)(\beta-2)p^{\beta-3}\).
Insert this in the exact truncation identity
\((F_h\bm g)_i-(Fg)(p_i)=\tfrac{h^{2}}{6}  p_i  g'''(\xi_i)\).
Since \(\xi_i\in(p_{i-1},p_{i+1})\subset[h,1]\), \( \xi_i \leq p_{i+1}\), we get the error bound directly. Since $\beta\geq 3$, the error at the boundary is also $\co(h^2)$
\end{proof}

%------------------------------------------------------------
\subsection{Pointwise error at time $T$}
\label{sec:time-T-error}

We consider applying the finite difference appproximation \eqref{eq: Fh2} with  a special class of initial conditions for \cref{eq: pde'}, 
\begin{equation}\label{g-init}
    g(p) := \frac{1}{Z_\beta} p^{\beta},\qquad
\beta=\frac1\theta -\frac12>0, \qquad 0<\theta <2,
\end{equation}
with appropriate values of $a$ to be determined. In particular, we have 
\( F g = g\), which fulfills the moment conditions. Here we analyze the error due to the approximation of $F$ by \eqref{eq: Fh2}.

For the discrete approximation, we choose the following approximation for $\kee{r}$,
\begin{equation}\label{eq:g2rh}
    \ket{r_h}= \sum_{j=0}^M g(p_j) \ket{j}.
\end{equation}

The constant $Z_\beta$ is selected so that,
\[
  \sum_{j=0}^M
 g(p_j)^2 =1. \]
A direct calculation shows that,
\begin{equation}\label{Z-beta}
    Z_\beta = \sqrt{\frac{M}{2\beta+1}} + \co\left(\frac{1}{M}\right).
\end{equation}

Let $\bm g$ be a vector with components \(g_i:=g(p_i)\) as the grid values.
Denote by   $u(t,p)$ the solution of \cref{eq: pde} with initial value in \cref{g-init}, and as its approximation, we denote
 by  $\bm u(t)$ the solution of 
 \begin{equation}\label{num-u}
     \dot{\bm u}=- \kappa(t)  \theta F_h \bm u, \quad \bm u(0)= \bm g.
 \end{equation}
In addition, we impose the exact boundary conditions,
\begin{equation}\label{eq:exact-bc}
        u_M(t)= u(t,1).
\end{equation}

To understand how the method might be implemented, we recall the definition of the difference operator $F_h$ \eqref{eq: Fh2}, and then write \cref{num-u} in a componentwise form $\forall 2<  i <M $,
\begin{equation}
    \dot{u}_i = -\frac{\theta  \kappa(t) }{4h} \left( (p_{i} + p_{i+1})  u_{i+1} - ( p_{i-1} + p_{i}) u_{i-1} \right). 
\end{equation}
When $i=M-1$, the boundary condition \eqref{eq:exact-bc} kicks in. Meanwhile, at the left boundary, \cref{eq: Fh2} leads to the ODE,
\[
\begin{aligned}
  \dot{u}_1 =& - \theta  \kappa(t) \big( \frac34 u _{2}-\frac{1}{4\sqrt2} u_{0} \big)\\
  \dot{u}_0 =&  - \theta  \kappa(t) \frac{1}{2\sqrt2} u_{1}.
\end{aligned}
\]
Therefore, no boundary condition is needed at the left boundary. This is consistent with the fact that the exact solution propagates to the left. Therefore, the discretization respects the direction of the wave propagation and can be regarded as an upwind scheme.

\medskip 

To derive an error bound, we write the pointwise error \(e_i(t):=u_i(t)-u   \bigl(t,p_i\bigr)\). In light of \cref{num-u} and \cref{eq:exact-bc}, we have $e_M(t)=0, $ and $\bm e(0)=0$.

\medskip
\begin{lemma}[Error after time $T$]
\label{lem:T-error}
The solution  of the PDE \eqref{eq: pde'} with initial value \eqref{g-init} is given by,
  \[ u(t,p)= y(t) g(p), \]
  where $\dot{y}= - \kappa(t) y $ and $y(0)=1.$
Assume that $\beta = 1/\theta  - 1/2 \geq 3$.  Then for every grid point $p_i$, we have
\begin{equation}
    |e_i(T)| \le 
{C(\theta )}  T  h^2,
\quad
C(\theta )= \frac{\theta }{6} \beta(\beta-1)(\beta-2).
\end{equation}
\end{lemma}

\begin{proof}
We prove  the first part by direct verification:
\[
\frac{d}{dt} y(t) g(p) = - \kappa(t)  y(t) g(p) = -\theta  \kappa(t) y(t) F g(p) = -\theta  \kappa(t) F u.
\]
Here we have used the fact that $ \theta F g = g$.

To estimate the error, we first subtract the two evolution equations to obtain an error equation, 
\begin{equation}\label{error-eq}
    \dot{\bm e}(t) = - \kappa(t) \theta F_h \bm e(t) + \bm \eta(t),
\end{equation}
where 
\[
\eta_j(t):=-\theta \kappa(t) \bigl(F_hg - Fg\bigr) u(t, p_j),  
\]
is the local consistency error.  \cref{local-error-g} gives a time–uniform bound
\(
\|\eta(t)\|_{\infty}\le  K_{\max}  C(\theta ) h^{2} y(t),
\)
provided that $\beta\geq 3$.

Importantly, since $\bm e$ satisfies the zero boundary condition, from \cref{F-herm}, $F_h$ is skew Hermitian, thus $- \kappa(t) F_h $ in \cref{error-eq} generates some unitary operator $U(t,s)$.  In addition,  the variation of constant formula implies that
\[
\bm e(t)=\int_{0}^{t} U(t,s)\bm \eta(s)  ds.
\]

As a result, 
\[
\abs{e_i(t)} \le  C(\theta )  h^{2} \int_0^t e^{-\int_0^{t_1} \kappa(t_2) dt_2 } dt_1 \leq   C(\theta )  K_{\max}  h^{2} t.
\]
Using $t=T$ yields the claimed estimate.

\end{proof}

\subsection{Dilation using Eigen-functions $p^\beta$ with $\beta>0$}

The dilation method based on the differential operator $F$ in \cref{eq:F-cont} can also use an eigenfunction $\kee{r} \propto p^\beta$ with $\beta>0$. After the discretization, the residual error $\theta F_h \ket{r_h} - \ket{r_h} $ might not be zero in the interior. Therefore, it generates a local error that needs to be taken into account. Its impact on the final error is summarized in the following theorem.

\begin{theorem}\label{thm:second-order-error}
Suppose $\beta\ge 3$ and discretize $F$ by \eqref{eq:Fh} on a uniform grid.
Choose $\theta$ so that the light-cone condition \eqref{a-CFL} holds. Then the approximate
dilation obeys
\begin{equation}\label{error-2}
\left\|
e^{\int_0^T L(s) ds} |\bm x_0\rangle
- (\langle l_h|\otimes I)  U_h(T,0) \bigl(|r_h\rangle \otimes |\bm x_0\rangle\bigr)
\right\|
 \le 
\frac{ C(\theta)  X(T)  K_{\max} T  h^{2} }{Z_\beta}
 + 
\bigl\|(\bra{i} \otimes I) \chi(T)\bigr\|,
\end{equation}
where $U_h(T,0)$ is the unitary generated by the discretized dilated Hamiltonian,
$X(T):=\sup_{t\in[0,T]}\|\bm x(t)\|$, and
\[
C(\theta) = \frac{\theta}{6} \beta(\beta-1)(\beta-2),
\qquad
\beta=\frac{1}{\theta}-\frac12.
\]
\end{theorem}

In \cref{a-CFL,error-2} one sees the trade-off in $\theta$: making $\theta$ smaller
improves the light-cone margin (delays boundary influence) but increases the bulk prefactor
$C(\theta)$; indeed $C(\theta)\sim (6\theta^2)^{-1}$ as $\theta\to0$.

\begin{proof}
    Our goal is to implement the following dilation formula,
\begin{equation}
            ( I_A \otimes \brr{l} )  
       \ct e^{-i \int_0^t\widetilde H(s) ds}   
        (I_A \otimes \kee{r} )
        =
        \ct e^{\int_0^t L(s)ds}, 
        \qquad
        \forall  t\ge 0,
\end{equation}
where,
\begin{equation}\label{dH}
       \Tilde{H}:= I_A \otimes H 
         +
         i F_\theta \otimes K. 
   \end{equation}

We first present the analysis for a general initial condition $\ket{g}$ is defined in \cref{g-init}. For the case when $\beta$ is an integer, a much easier analysis will be presented in the next section.   To connect to the analysis in the preceding sections, we let the exact solution and approximate solution from the dilated dynamics be $\ket{\Psi}$ and $\ket{\Phi}$, respectively. We compare the respective equations as follows,
\begin{equation}
    \begin{aligned}
        \frac{d}{dt} \ket{\Psi(t)} =& -i \left( I_A \otimes H 
         + iF_\theta \otimes K(t) \right) \ket{\Psi(t)}, \quad&  \ket{\Psi(0)} = \ket{g}\ket{\bm x_0},  \\
        \frac{d}{dt} \ket{\Phi(t)} =&-i \left( I_A \otimes H 
         +i \theta F_h \otimes K(t) \right)\ket{\Phi(t)}, \quad&  \ket{\Phi(0)} = \ket{r_h}\ket{\bm x_0}.  
    \end{aligned}
\end{equation}
Here $\ket{g}$ is defined in \cref{g-init} and $\ket{r_h}$ is the discrete approximation in \cref{eq:g2rh} for the consistency with the continuous problem. A similar calculation shows that \[\Psi(t)= \ket{r_h} \ct e^{\int_0^t \big(-iH(s) + K(s) \big)ds} \ket{\bm x_0},  \] thus encoding the correct solution.

Let us pinpoint the sources of error 
\begin{enumerate}
    \item the discretization of $F_\theta$ by $\theta F_h$.
    \item the zero boundary condition at $p_M=1$ imposed on $\chi$ to ensure the Hermitian property of $F_h$.
\end{enumerate}

To analyze the sources of the error, we introduce function $\chi,$ such that,
\begin{equation}\label{chi}
\frac{d}{dt} \ket{\chi(t)} = -i\left( I_A \otimes H 
         +i \theta F_h \otimes K(t) \right)\ket{\chi(t)}, \quad  \ket{\chi(0)} = \ket{r_h}\ket{\bm x_0},
\end{equation}
together the the boundary condition that $\bigl(\bra{M}\otimes I\bigr) \chi(t)=  \bigl(\bra{M}\otimes I\bigr)\Psi(t).$

Therefore, the total error is decomposed into,
\begin{equation}\label{err-decmp}
    \ket{\Phi(t)} -  \ket{\Psi(t)} = \underbrace{\ket{\chi(t)} -  \ket{\Psi(t)} }_{\bm e(t)} + \underbrace{\ket{\Phi(t)} -  \ket{\chi(t)} }_{\bm z(t)}
\end{equation}

First, we observe that the error equation \eqref{error-eq}, by extending the definite $\bm e(t)$ to
 \[
  \bm e(t)= \ket{\Phi(t)} - \ket{\Psi(t)},
 \]
simplify becomes,
\begin{equation}\label{error-eq0phipsi}
     \dot{\bm e}(t) = -i \left( I_A \otimes H 
         + i \theta F_h \otimes K(t) \right) \bm e(t) + \bm \eta(t),
\end{equation}
where the local truncation error becomes
\begin{equation}
    \bm \eta(t):=-\theta \bigl(F_hg - Fg\bigr) \otimes \ct e^{\int_0^t \big(-iH(s) + K(s) \big) ds} \ket{\bm x_0}.
\end{equation}

Let $Q(T)= \max_{t\in[0,T]} \norm{\ct e^{\int_0^t \big(-iH(s) + K(s) \big) ds} \ket{\bm x_0}  }$. 
Using the Hermitian property of $F_h$ and the dissipative condition,  we arrive at the same bound that $ \norm{\bm e(t)} \leq   C(\theta )  h^{2} Q(T) t$  as in the previous section.

Meanwhile, the other part of the error is,
\[
  \bm z(t)=  \ket{\Phi(t)} - \ket{\chi(t)},
\]
which satisfies the equation,
\begin{equation}\label{error-z}
     \frac{d}{dt}\bm z(t) = -i\left( I_A \otimes H 
         + \theta F_h \otimes K(t) \right) \bm z(t), 
\end{equation}
with zero initial condition $\bm z(t)=0$ while subject to the boundary condition, 
\begin{equation}\label{bc-z}
    \bigl(\bra{M}\otimes I\bigr) z(t) = - \braket{M}{r_h}  \ct e^{\int_0^t \big(-iH(s) + K(s) \big) ds} \ket{\bm x_0}.
\end{equation}
We note that $\braket{M}{r_h}= 1/Z_\beta= \co\left(\sqrt{\frac{\beta}{M}} \right)$ from \cref{Z-beta}

\end{proof}

\medskip

Let $H_{\max}=\max_{0\le t\le T}\|H(t)\|$ and
$K_{\max}=\max_{0\le t\le T}\|K(t)\|$. Note that $K_{\max}$ is a property of the
system and does not scale with $M$ (whereas $\|F_h\|=\Theta(M)$).
Combining \eqref{error-2} with time-dependent interaction-picture simulation
\cite{low2018hamiltonian} gives our first main query bound (with
$\widetilde{\mathcal O}$ hiding polylogarithmic factors).

\begin{theorem}\label{thm:first-query}
 The dilated dynamics can be simulated to precision $\epsilon$ using
\[
\widetilde{\mathcal O} \left(
        T H_{\max}
         + 
        \frac{X(T)^{1/2} T^{3/2} K_{\max}}{\epsilon^{1/2}}
\right)
\]
block-encoding queries to the oracles for $H(t)$ and $K(t)$.
\end{theorem}

The $\sqrt{T/\epsilon}$ in the second term reflects the second-order truncation error
in \eqref{eq:Fh}. When $K_{\max}\ll H_{\max}$, this already yields reasonable complexity.
A remaining subtlety is the postselection success probability. With the weighted
discretization
\[
\ket{r_h}
=
Z_\beta^{-1}\sum_{j=0}^{M} p_j^{\beta} w_j^{1/2} \ket{j},
\qquad
Z_\beta^2
=
\sum_{j=0}^{M} w_j p_j^{2\beta}
 \approx 
\int_0^1 p^{2\beta} dp
=\frac{1}{2\beta+1},
\]
postselecting a single site $p_\ast=p_i$ on a uniform grid ($w_i\approx 1/M$) gives
\begin{equation}\label{eq:rstar}
  P_{\text{succ}}
   =  |\langle i|r_h\rangle|^2
   \approx  \frac{(2\beta+1) p_\ast^{2\beta}}{M}.
\end{equation}
Thus when $K_{\max}T\gg1$ (stiff regimes) and $\beta$ is large, the success probability
for a fixed site can become exponentially small in $\beta$; in practice one should
choose $j_\ast$ in the interior (or aggregate several sites) to maintain a reasonable
postselection rate.

\end{document}